\RequirePackage{fix-cm}
\documentclass[smallextended]{svjour3}       
\smartqed  
\usepackage{graphicx}
\usepackage{amsmath, amsfonts, amssymb, amscd,enumerate}
\usepackage[mathcal]{euscript}

\begin{document}

\title{A Model System of Mixed Ionized  Gas Dynamics}


\author{ASAKURA, Fumioki}


\institute{Professor Emeritus, Osaka Electro-Communication Univ., \at
              18-8 Hatsucho, Neyagawa, Osaka 572, JAPAN \\
              Tel.: +72-824-1131\\
              Fax: +72-824-0014\\
              \email{asakura@osakac.ac.jp}
}

\date{Received: date / Accepted: date}

\maketitle

\begin{abstract}
The aim of this paper is to study a one dimensional model system of equations for ionized gas dynamics at high temperature where the gas is a mixture of two kinds of monatomic gas. In addition to the mass density, pressure, temperature and particle velocity, degrees of ionization of  both gases are also involved. By assuming that the local thermal equilibrium is attained, Saha's ionization equations are added. Thus the equations are supplemented by  the first and second law of thermodynamics, a single equation of state and, in addition, a set of  thermodynamic equations. 
\par
The equations constitute a strictly hyperbolic system, which guarantees that the initial value problem is well-posed locally in time for sufficiently smooth initial data. However the geometric properties of the system are rather complicated: in particular, we prove the existence of a region where convexity (genuine nonlinearity) fails for forward and backward characteristic fields. Also we study thermodynamic properties of shock waves by a detailed analysis of the Hugoniot locus, which is used in a mathematical study of existence and uniqueness of solutions to the shock tube problem.
\medskip
\par
\noindent
\textit{2010~Mathematics Subject Classification:} 35L65, 35L67, 76N15
\end{abstract}
\section{Introduction}\label{sec:Introduction}
A shock wave is a propagating discontinuity of density, pressure, temperature and etc., which is supersonic with respect to the gaseous medium ahead of it and subsonic with respect to that behind it. Behind a shock wave, not only pressure but also temperature increases abruptly and the gas is heated to high temperatures. Strong shock waves are obtained and replicated by shock tube operations under ordinary circumstances. Hence the shock tube is a convenient and widely used device for obtaining high temperature gases in the laboratory. The shock front and the state behind it in the shock tube are determined by the state ahead of it and the speed of driving gas, which is a mathematical problem called the {\it shock tube problem\/} in this paper.
\par
    When the gas behind the shock front is heated to a high temperature, almost all molecules become dissociated and finally  its atoms become partially ionized: ${\rm X} \rightleftarrows {\rm X}^+ + {\rm e}^-.$ 
Numerous spectroscopic measurements of atomic parameters and thermodynamic equilibrium of plasma  thus generated  have been done, for example, in various Helium-Hydrogen mixtures (\cite{Fukuda_2},\cite{Fukuda_3}).
The model system of mixed ionized gas dynamics that we discuss in this paper is proposed by Fukuda-Okasaka-Fujimoto in \cite{Fukuda-Okasaka-Fujimoto}\footnote{An English translation of \cite{Fukuda-Okasaka-Fujimoto} is available upon request to F. Asakura.} for the purpose of providing a theoretical basis for their observations. The system consists of equations of  macroscopic motion for 1-d mixed gas dynamics. Its particular nature is: degree of ionization of each gas is considered to be a thermodynamic variable.
\par
The present paper is a continuation of \cite{Asakura-Corli_ionized}, \cite{Asakura-Corli_reflected}, \cite{Asakura-Corli_RIMS} and our aim is to perform mathematical analysis for the model system and show its basic thermodynamic properties. To the best of our knowledge such a study has never been done previously while the system of gas dynamics attracted the interest of several researchers in the last decade, however mostly for ideal gases \cite{Smoller}; we quote \cite{Menikoff-Plohr} for the case of real gases. 
For a single monatomic ionized gas, studies have been done in  \cite{Asakura-Corli_ionized}, \cite{Asakura-Corli_reflected}, \cite{Asakura-Corli_RIMS}. 
\par
Basic thermodynamic variables are denoted in this paper by $T:$ temperature, $p:$ pressure, $\rho:$ mass density,  $v= 1/\rho:$ specific volume, $e:$ specific internal energy and $S:$ specific entropy. The flow velocity is denoted by $u$ and the  (specific) total energy by $\mathcal{E} = \frac{1}{2}u^2 + e.$ 
The  system of equations of one-dimensional motion for gas dynamics consists of the following three conservation laws: conservation of  mass, momentum and  energy
\begin{equation}\label{eq:system}
\left\{
\begin{array}{l}
\rho_t + (\rho u)_x = 0,
\\
\displaystyle (\rho u)_t + (\rho u^2 + p)_x = 0,\rule{0ex}{2.75ex}
\\
\displaystyle\left(\rho \mathcal{E}\right)_t + \left(\rho  \mathcal{E}u + pu\right)_x = 0,\rule{0ex}{2.75ex}
\end{array}
\right.
\end{equation}
which are supplemented by  the first and second law of thermodynamics
\begin{equation}\label{eq:second-principle}
de  = TdS - pdv,
\end{equation}
a single equation of state and a set of  thermodynamic equations. 
For brevity we will refer to $S,$  $e$ and $\mathcal{E}$ as the entropy, internal  and total energy, respectively.
\par
For a partially ionized  {\it single\/} monatomic gas, let $n_{\rm a},$  $n_{\rm i}$ and $n_{\rm e}$ denote, respectively, the concentration (number per unit volume) of atoms, ions and electrons.  The equation of state depends on the {\it degree of ionization\/} $\alpha = \frac{n_{\rm e}}{n_{\rm a} + n_{\rm i}}$ having the  form 
      \begin{equation}\label{eq:pressure-ion_single}
        p = \frac{R}{M}\rho T(1 + \alpha)
      \end{equation}
      where $R$ denotes the {\it universal gas constant\/} and $M$ the {\it molar mass} of the monatomic gas (\cite{Fukuda-Okasaka-Fujimoto}). This  model system is similar to the system of an ideal dissociating diatomic gas studied by Lighthill  in \cite{Lighthill_1}.
      \par
It is found that, at any given high temperature $T$ and volume $V,$  these ionization reactions reach a state of equilibrium which is analogous to the chemical equilibrium for usual chemical reactions  whose equilibrium condition is the law of mass action: the ratio
\(
\frac{n_{\rm i}n_{\rm e}}{n_{\rm a}} 
\)
depends only on the temperature $T$.
 An actual formula was derived   by M. Saha  in \cite{Saha}, namely,
\begin{equation}\label{eq:saha}
\frac{n_{\rm i} n_{\rm e}}{n_{\rm a}} = \frac{2G_{\rm i}}{G_{\rm a}} \frac{(2\pi m_{\rm e} kT)^{\frac{3}{2}}}{h^3}\, e^{-\frac{T_{\rm i}}{T}},
\end{equation}
see also \cite{Bradt},\cite{Fermi},\cite{Vincenti-Kruger},\cite{Zeldovich-Raizer}. Here we denote  the partition functions of the neutral state and of the $1$-ionized state by $G_{\rm a}$ and $G_{\rm i},$ respectively; $m_{\rm e}$ is the electron mass, $k$ the Boltzmann constant, $h$ the Planck constant and
$T_{\rm i}=\frac{\chi}{k}$ the ionization energy measured by the temperature, where $\chi$ is the  first ionization potential \cite[\S 5, (4.8)]{Vincenti-Kruger}. On the other hand,  Saha's law  \eqref{eq:saha} is written as
\begin{equation}\label{eq:saha2}
\frac{\alpha^2}{1-\alpha^2}
= \frac{2G_{\rm i}}{G_{\rm a}} \frac{(2\pi m_{\rm e})^{\frac{3}{2}}(kT)^\frac52}{ph^3}\, e^{-\frac{T_{\rm i}}{T}},
\end{equation}
(see  \cite[(209)]{Fermi}, \cite[\S V.4, (4.9)]{Vincenti-Kruger}), showing that $\alpha$ can be regarded as a thermodynamic variable.
\par
Since the electric intermolecular process of ionization occurs much faster then the fluid-dynamic phenomenon of shock formation, see for instance \cite[VII, \S 11]{Zeldovich-Raizer} and \cite[VII, \S 10, Table 7.3]{Zeldovich-Raizer}, we may assume that a local thermodynamic equilibrium is everywhere attained: that is,  Saha's law \eqref{eq:saha2} holds everywhere even in presence of shock waves, which is one of the postulates of the present model system of ionized gas dynamics. Thus the equatios \eqref{eq:pressure-ion_single} and \eqref{eq:saha2} constitute the equation of state and a thermodynamic equation. 
\par
Now let us consider one mole mixture of monatomic gases A and B. The ionization reactions are represented as
$
{\rm A} \rightleftarrows {\rm A}^+ + {\rm e}^-, \quad
{\rm B} \rightleftarrows  {\rm B}^+ + {\rm e}^-.
$
We denote the number of atoms and ions for each gas by $N_{\rm a}^{\rm A},N_{\rm a}^{\rm B}$ and $N_{\rm i}^{\rm A}, N_{\rm i}^{\rm B},$ respectively. The number of electrons are denoted by $N_{\rm e}.$ Note that
$$
N_{\rm a}^{\rm A} + N_{\rm i}^{\rm A} + N_{\rm a}^{\rm B} + N_{\rm i}^{\rm B} = N_0:\ \text{Avogadro Number}, \quad N_{\rm e} = N_{\rm i}^{\rm A} + N_{\rm i}^{\rm B}.
$$
The concentration of atoms, ions and electrons are defined by 
$n_{\rm a}^{\rm A} = \frac{N_{\rm a}^{\rm A}}{V}, n_{\rm a}^{\rm B} = \frac{N_{\rm a}^{\rm B}}{V}, n_{\rm i}^{\rm A} = \frac{N_{\rm i}^{\rm A}}{V}, n_{\rm i}^{\rm B} = \frac{N_{\rm i}^{\rm B}}{V}, n_{\rm e} = \frac{N_{\rm e}}{V} ,$ respectively.
\par
By denoting  $G_{\rm a}^{\rm A} ,  G_{\rm a}^{\rm B}:$ the partition functions of the neutral state, $G_{\rm i}^{\rm A},  G_{\rm i}^{\rm B}:$ same for the $1$-ionized state, and $\chi^{\rm A}, \chi^{\rm B}:$ first ionization potentials,  the coupled Saha's laws for mixed monatomic gas are presented as the following.
\begin{equation}\label{eq:coupled saha}
   \frac{n_{\rm i}^{\rm A}n_{\rm e}}{n_{\rm a}^{\rm A}} = \frac{2G_{\rm i}^{\rm A}}{G_{\rm a}^{\rm A}}\left(\frac{{2\pi m_{\rm e} kT}}{h^2}\right)^{\frac{3}{2}}e^{-\frac{\chi^{\rm A}}{kT}}, \quad
     \frac{n_{\rm i}^{\rm B}n_{\rm e}}{n_{\rm a}^{\rm B}} = \frac{2G_{\rm i}^{\rm B}}{G_{\rm a}^{\rm B}}\left(\frac{{2\pi m_{\rm e} kT}}{h^2}\right)^{\frac{3}{2}}e^{-\frac{\chi^{\rm B}}{kT}}
\end{equation}
For A: hydrogen atom A, we have $\chi^{\rm A} = 13.59844$ eV and for B: helium atom B, $\chi^{\rm B} = 24.58741$ eV.  First ionization temperatures are
$$
   T_{\rm A} = \frac{\chi^{\rm A}}{k} = 1.5780 \times 10^5, \quad T_{\rm B} = \frac{\chi^{\rm B}}{k}= 2.8532 \times 10^5.
$$
Note that
$T_{\rm A} <  T_{\rm B} < 2T_{\rm A}.$ 
We have also 
$
     \frac{2G_{\rm i}^{\rm A}}{G_{\rm a}^{\rm A}} = 1, \,  \frac{2G_{\rm i}^{\rm B}}{G_{\rm a}^{\rm B}} = 4.
$
 We will assume that a local thermodynamic equilibrium is everywhere attained: that is the  coupled Saha's laws \eqref{eq:coupled saha} hold everywhere even in presence of shock waves.
\par
The present  model system is constructed on the basis of several postulates that we now expose in detail. 
By denoting  the Debye radius \cite[III-2-\S11]{Zeldovich-Raizer} by  $\lambda_D, $ the {\it plasma parameter}: $\Lambda = \frac{4}{3}\pi n_{\rm e} \lambda_D^3$ is a dimensionless number defined by the number of electrons in a Debye sphere. 
We recall that the ratio between the potential energy and the kinetic energy is of order $\Lambda^{-\frac{2}{3}}.$
We first assume that:
\begin{itemize}
\item Gas is a mixture of monatomic gases A and B satisfying  $T_{\rm A} < T_{\rm B} \leq 2T_{\rm A}$
\item All collisions are perfectly elastic (or effects of collisions among the particles can be neglected)
   \item Gravitational effects, viscosity and thermal conductivities are disregarded
\item  $\Lambda \gg 1$, which means that the interaction potential energies of the charged particles are negligible with respect to the kinetic energies, and electrostatic interactions are relatively rare
\item Local thermodynamic equilibrium is everywhere attained
\end{itemize}
The fourth postulate above is motivated by the high temperatures considered in \cite{Fukuda-Okasaka-Fujimoto}, see \cite[\S 78]{Landau-Lifshitz_SP}, \cite[(3.77)]{Zeldovich-Raizer} for further details.
\par
The {\it degree of ionization} and {\it fraction} for each gas is defined by
$$
\textstyle 
\alpha_{\rm A} = \frac{n_{\rm i}^{\rm A}}{n_{\rm a}^{\rm A} + n_{\rm i}^{\rm A}}, \quad
\alpha_{\rm B} = \frac{n_{\rm i}^{\rm B}}{n_{\rm a}^{\rm B} + n_{\rm i}^{\rm B}},
$$
$$
\textstyle
\beta = \frac{N_{\rm a}^{\rm A} + N_{\rm i}^{\rm A}}{N_0} = \frac{n_{\rm a}^{\rm A} + n_{\rm i}^{\rm A}}{n_0}, \quad
1 - \beta = \frac{N_{\rm a}^{\rm B} + N_{\rm i}^{\rm B}}{N_0} = \frac{n_{\rm a}^{\rm B} + n_{\rm i}^{\rm B}}{n_0}.
$$
The density and molar mass of each gas are denoted by $\rho_{\rm A},\, \rho_{\rm B}$ and $M_{\rm A},\, M_{\rm B},$ respectively.
The pressure is a sum of partial pressures with respect to atoms, ions and electrons:
$$
p = p_{\rm a} + p_{\rm i} + p_{\rm e} = p_{\rm a} + 2p_{\rm i},\quad p_{\rm j} = kn_{\rm j}T\ ({\rm j} = {\rm a,\, i,\, e}).
$$
Then  by setting $\alpha = \beta \alpha_{\rm A} + (1 - \beta)\alpha_{\rm B}$ 
\begin{align*}
  p &= p_{\rm a}^{\rm A} + 2 p_{\rm i}^{\rm A} + p_{\rm a}^{\rm B} + 2 p_{\rm i}^{\rm B}
  = k\left(n_{\rm a}^{\rm A} + 2 n_{\rm i}^{\rm A} + n_{\rm a}^{\rm B} + 2 n_{\rm i}^{\rm B}\right)T\\
  &= k\left[\left(n_{\rm a}^{\rm A} + n_{\rm i}^{\rm A}\right)\left(1 + \alpha_{\rm A}\right) + \left(n_{\rm a}^{\rm B} + n_{\rm i}^{\rm B}\right)\left(1 + \alpha_{\rm B}\right)\right]T
= kn_0\left(1 + \alpha\right)T.
\end{align*}
By noticing $1 + \frac{n_{\rm a}^{\rm A} }{n_{\rm i}^{\rm A}} = \frac{1}{\alpha^{\rm A}},1 + \frac{n_{\rm a}^{\rm B} }{n_{\rm i}^{\rm B}} = \frac{1}{\alpha^{\rm B}},$ Saha's laws take the forms
\begin{align*}
  \frac{n_{\rm i}^{\rm A}n_{\rm e}}{n_{\rm a}^{\rm A}}
  &= \left(\frac{\alpha_{\rm A}}{1 - \alpha_{\rm A}}\right)\left(n_{\rm i}^{\rm A} + n_{\rm i}^{\rm B}\right)  
  =\frac{n_0\alpha_{\rm A}\alpha}{1 - \alpha_{\rm A}}
   = \frac{2G_{\rm i}^{\rm A}}{G_{\rm a}^{\rm A}}\left(\frac{{2\pi m_{\rm e} kT}}{h^2}\right)^{\frac{3}{2}}e^{-\frac{T_{\rm A}}{T}}, \\
  \frac{n_{\rm i}^{\rm B}n_{\rm e}}{n_{\rm a}^{\rm B}}
  &= \left(\frac{\alpha_{\rm B}}{1 - \alpha_{\rm B}}\right)\left(n_{\rm i}^{\rm A} + n_{\rm i}^{\rm B}\right)
  = \frac{n_0\alpha_{\rm B}\alpha}{1 - \alpha_{\rm B}}
  = \frac{2G_{\rm i}^{\rm B}}{G_{\rm a}^{\rm B}}\left(\frac{{2\pi m_{\rm e} kT}}{h^2}\right)^{\frac{3}{2}}e^{-\frac{T_{\rm B}}{T}}.
\end{align*} 
Thus we conclude that thermodynamic equations have the forms
\begin{align}
  p &= \frac{(1 - \alpha_{\rm A})(1 + \alpha)}{\alpha_{\rm A}\alpha}\frac{2G_{\rm i}^{\rm A}}{G_{\rm a}^{\rm A}}\frac{(2\pi m_{\rm e})^{\frac{3}{2}}(kT)^{\frac{5}{2}}}{h^3}e^{-\frac{T_{\rm A}}{T}}\nonumber\\
    &= \frac{(1 - \alpha_{\rm B})(1 + \alpha)}{\alpha_{\rm B}\alpha}\frac{2G_{\rm i}^{\rm B}}{G_{\rm a}^{\rm B}}\frac{(2\pi m_{\rm e})^{\frac{3}{2}}(kT)^{\frac{5}{2}}}{h^3}e^{-\frac{T_{\rm B}}{T}}.\label{eq:p}
\end{align}
Also we have a compatibility condition
\begin{equation}\label{eq:compatibility_1}
    \frac{2G_{\rm i}^{\rm A}}{G_{\rm a}^{\rm A}}\frac{1 - \alpha_{\rm A}}{\alpha_{\rm A}}e^{-\frac{T_{\rm A}}{T}}\\
    = \frac{2G_{\rm i}^{\rm B}}{G_{\rm a}^{\rm B}}\frac{1 - \alpha_{\rm B}}{\alpha_{\rm B}}e^{-\frac{T_{\rm B}}{T}}.
\end{equation}

\par
 We next assume that:
\begin{itemize}
    \item Gases are well mixed so that: 
$
 \rho = \beta \rho_{\rm A} + (1 - \beta) \rho_{\rm B};
$
   \item Pressure of each gas takes the form 
$$
  p_{\rm A} = \dfrac{R\rho_{\rm A}}{M_{\rm A}} T(1 + \alpha_{\rm A}), \quad
  p_{\rm B} = \dfrac{R\rho_{\rm B}}{M_{\rm B}} T(1 + \alpha_{\rm B})
$$
  \item Specific enthalpies are defined by
$$
 h_{\rm A} = \frac{5R}{2M_{\rm A}} T(1 + \alpha_{\rm A}) + \frac{RT_{\rm A}}{M_{\rm A}}\alpha_{\rm A},\quad
 h_{\rm B} = \frac{5R}{2M_{\rm B}} T(1 + \alpha_{\rm B}) + \frac{RT_{\rm B}}{M_{\rm B}}\alpha_{\rm B}
$$
 \item Macroscopic motion of the gas flow is one-dimensional
\end{itemize}
We deduce from the above assumptions that the total pressure is 
\begin{align*}
  p &= \beta p_{\rm A} + (1 - \beta) p_{\rm B}
    = \beta\frac{R\rho_{\rm A}}{M_{\rm A}} T(1 + \alpha_{\rm A}) 
    + (1 - \beta)\frac{R\rho_{\rm B}}{M_{\rm B}} T(1 + \alpha_{\rm B}).
\end{align*}
Thus
\begin{align*}
  \frac{p}{\rho} 
    &= \frac{\beta\frac{R}{V} T(1 + \alpha_{\rm A}) 
    + (1 - \beta)\frac{R}{V} T(1 + \alpha_{\rm B})}
{\beta \frac{M_{\rm A}}{V} + (1 - \beta) \frac{M_{\rm B}}{V}}
  = \frac{RT\left[1 + \beta\alpha_{\rm A} + (1 - \beta)\alpha_{\rm B}\right]}
{\beta M_{\rm A} + (1 - \beta) M_{\rm B}}.
\end{align*}
Denoting $\alpha = \beta\alpha_{\rm A} + (1 - \beta)\alpha_{\rm B}$ and $M = \beta M_{\rm A} + (1 - \beta) M_{\rm B},$ we obtain
      \begin{equation}\label{eq:pressure-ion}
        p = \frac{R}{M}\rho T(1 + \alpha)
      \end{equation}
 which is the  equation of state. The total specific enthalpy is 
\begin{equation}\label{eq:total enthalpy}
  h = \frac{\beta M_{\rm A}h_{\rm A} + (1 - \beta)M_{\rm B}h_{\rm B} }{\beta M_{\rm A} + (1 - \beta)M_{\rm B}}
  = \frac{5RT}{2M}(1 + \alpha)
+ \frac{R}{M}\left[\beta T_{\rm A}\alpha_{\rm A} + (1 - \beta) T_{\rm B}\alpha_{\rm B}\right]
 \end{equation}
\par
After a short review of basic thermodynamics, we show some basic calculus lemmas in Section \ref{sec:basic-thermodynamics}. The physical entropy functions are constructed in Section \ref{sec:entropy construction}. We show that system \eqref{eq:system} is strictly hyperbolic and compute characteristic fields in Section \ref{sec:Equations Gas Dynamics}. However, unlike the ideal polytropic case, the forward and backward characteristic fields of the system are not genuinely nonlinear and we study the set where this happens in Section \ref{sec:convexity}. We refer to \cite{Dafermos}, \cite{Smoller} for more information on systems of conservation laws.
We study in Section \ref{sec:compatibility} the relation between $\alpha_{\rm A}$ and $\alpha_{\rm B}.$ 
A detailed study of the Hugoniot locus of the system is carried out in Section \ref{sec:Hugoniot}. Though Hugoniot loci are monotone in $(T, \alpha)$-plane in a single monatomic case, they are not always monotone in the present mixed monatomic case: If $\beta$ is sufficiently small, then they lose monotonicity at some base state. Thus the degree of ionization does not always increase across the shock front, even if the temperature increases. However we prove that the pressure actually increases as the temperature increases. In order to fit the mathematical data to ordinary circumstances, we propose an approximation of Hugoniot locus in Section \ref{sec:approximate Hugoniot}. We apply our results to the shock tube problem in Section \ref{sec:shock tube}. Basic results: existence and uniqueness are established, which provides  a rigorous mathematical basis to the physical phenomena observed in \cite{Fukuda-Okasaka-Fujimoto}. Behaviour of isentropes and detailed computations for the proof of uniqueness are shown in appendices.

\section{Basic Thermodynamics and Calculus Lemmas}\label{sec:basic-thermodynamics}
\setcounter{equation}{0}
 First we adopt  $p$ and $T$ as a set of independent thermodynamic state variables.
By introducing the enthalpy $h = e + pv,$ the first and second law 1-\eqref{eq:second-principle} becomes
$$
   dh = T\,dS + v\,dp =T\left(\frac{\partial S}{\partial T}\right)_{p} dT +  T\Bigl[\left(\frac{\partial S}{\partial p} \right)_{T} + v\Bigr]dp 
$$
As usual, a subscript as $T$ or $p$ above means that the derivative is computed by holding the subscripted variable  fixed. 
We also introduce the Gibbs function  
\(
g = h - TS
\), see \cite[(111)]{Fermi}, 
and we have  
\begin{equation}\label{eq:Gibbs-derivatives}
dg = v\, dp - S\,dT = \left(\frac{\partial g}{\partial p}\right)_Tdp 
+ \left(\frac{\partial g}{\partial T}\right)_p dT.
\end{equation}
\paragraph{Maxwell's Relations:} We deduce by \eqref{eq:Gibbs-derivatives} the compatibility condition
\begin{equation}\label{eq:Maxwell}
\left(\frac{\partial v}{\partial T}\right)_p = -\left(\frac{\partial S}{\partial p}\right)_T,
\end{equation}
which is one of  so-called Maxwell relations. 
In turn, by \eqref{eq:Gibbs-derivatives} and \eqref{eq:Maxwell} we obtain
$$
\left(\frac{\partial h}{\partial p} \right)_{T}
  = T\left(\frac{\partial S}{\partial p} \right)_{T} + v
   =   -T\left(\frac{\partial v}{\partial T} \right)_{p} + v, \quad
  \left(\frac{\partial h}{\partial T} \right)_{p}
 = T\left(\frac{\partial S}{\partial T}\right)_{p}.
$$
Thus we have the following proposition.
\begin{proposition}[$p,T:$ set of  independent variables]\label{prop:Maxwell}
$$
  \left(\frac{\partial S}{\partial p} \right)_{\! T}
  = - \left(\frac{\partial v}{\partial T} \right)_{\! p},\quad
  \left(\frac{\partial S}{\partial T}\right)_{\! p} = \frac{1}{T}\left(\frac{\partial h}{\partial T} \right)_{\! p}
$$
\end{proposition}
The specific volume $v$ is expressed  by 1-\eqref{eq:pressure-ion} 
as $
v = \frac{RT}{Mp}(1 + \alpha)
$
 and the enthalpy is 1-\eqref{eq:total enthalpy}.
The dimensionless entropy $\eta$ is defined by $\eta = \frac{M}{R} S.$ Consequently we have by Proposition \ref{prop:Maxwell}
\begin{lemma}\label{lem:d eta/dp, d eta/dT}
\begin{align}
  \left(\frac{\partial \eta}{\partial p} \right)_{\! T}
  &= - \frac{1}{p}\left[1 + \alpha + T \left(\frac{\partial \alpha}{\partial T} \right)_{\! p} \right] \label{eq:d eta/dp}\\
\left(\frac{\partial \eta}{\partial T} \right)_{\! p}
   &= \frac{5}{2T}(1 + \alpha)
 + \beta \left(\frac{5}{2} + \frac{T_{\rm A}}{T}\right)\left(\frac{\partial \alpha_{\rm A}}{\partial T} \right)_{\! p} + (1 - \beta)\left(\frac{5}{2} + \frac{T_{\rm A}}{T}\right)\left(\frac{\partial \alpha_{\rm B}}{\partial T} \right)_{\! p}\label{eq:d eta/dT}
\end{align}
\end{lemma}
\paragraph{Saha Equations:}
Setting 
$$
 T_{\rm A} = \frac{\chi^{\rm A}}{k}, \quad T_{\rm B} = \frac{\chi^{\rm B}}{k}, \quad \mu_{\rm A}^{-1} = \frac{2G_{\rm i}^{\rm A}}{G_{\rm a}^{\rm A}}\frac{(2\pi m_{\rm e})^{\frac{3}{2}}k^{\frac{5}{2}}}{h^3}, \quad
 \mu_{\rm B}^{-1} = \frac{2G_{\rm i}^{\rm B}}{G_{\rm a}^{\rm B}}\frac{(2\pi m_{\rm e})^{\frac{3}{2}}k^{\frac{5}{2}}}{h^3},
$$
we have  from  1-\eqref{eq:p} and  1-\eqref{eq:compatibility_1}
\begin{lemma}\label{lem:saha+compatibility}
Saha's equations take the forms
\begin{equation}\label{eq:saha bar(alpha)}
  \left(\frac{1}{\alpha_{\rm A}} - 1\right)\left(\frac{1}{\ \alpha\ } - 1\right) = \frac{\mu_{\rm A} p e^{\frac{T_{\rm A}}{T}}}{T^{\frac{5}{2}}}, \quad 
     \left(\frac{1}{\alpha_{\rm B}} - 1\right)\left(\frac{1}{\ \alpha\ } - 1\right) =\frac{\mu_{\rm B} p e^{\frac{T_{\rm B}}{T}}}{T^{\frac{5}{2}}}
\end{equation}
and the compatibility condition
\begin{equation}\label{eq:compatibility}
   \left(\frac{1}{\alpha_{\rm A}} - 1\right)\frac{e^{-\frac{T_{\rm A}}{T}}}{\mu_{\rm A}}
    = \left(\frac{1}{\alpha_{\rm B}} - 1\right)\frac{e^{-\frac{T_{\rm B}}{T}}}{\mu_{\rm B}}.
\end{equation}
\end{lemma}

\paragraph{Computation of $\left(\dfrac{\partial \alpha_{\rm A}}{\partial p} \right)_{T},\  \left(\dfrac{\partial \alpha_{\rm A}}{\partial T} \right)_{p}, \, \left(\dfrac{\partial \alpha_{\rm B}}{\partial p} \right)_{T},\  \left(\dfrac{\partial \alpha_{\rm B}}{\partial T} \right)_{p}:$ }
For the sake of brevity, we set
$$
q_{\rm A} = \alpha_{\rm A} (1 - \alpha_{\rm A}),\
q_{\rm B} = \alpha_{\rm B} (1 - \alpha_{\rm B}),\ q = \beta q_{\rm A} + (1 - \beta)q_{\rm B}.
$$
Differentiating Saha's equations, we have a system of Pfaff equations
\begin{align}
& \frac{\alpha(1 + \alpha) + \beta q_{\rm A} }{\alpha_{\rm A}^2 \alpha^2 } d\alpha_{\rm A} + \frac{(1 - \alpha_{\rm A})(1 - \beta)}{\alpha_{\rm A} \alpha^2} d\alpha_{\rm B}\nonumber\\
   &= -\frac{\mu_{\rm A}pe^{\frac{T_{\rm A}}{T}}}{T^{\frac{5}{2}}}\left[\frac{dp}{p} - \left(\frac{5}{2} + \frac{T_{\rm A}}{T}\right)\frac{dT}{T}\right],\label{eq:pfaff_1}\\
& \frac{(1 - \alpha_{\rm B})\beta}{\alpha_{\rm B} \alpha^2} d\alpha_{\rm A} + \frac{\alpha(1 + \alpha) + (1 -  \beta)q_{\rm B}}{\alpha_{\rm B}^2 \alpha^2} d\alpha_{\rm B}\nonumber\\
   &=  -\frac{\mu_{\rm B}pe^{\frac{T_{\rm B}}{T}}}{T^{\frac{5}{2}}}\left[\frac{dp}{p} - \left(\frac{5}{2} + \frac{T_{\rm B}}{T}\right)\frac{dT}{T}\right]\label{eq:pfaff_2}
\end{align}
which constitutes a system of linear equation of $d\alpha_{\rm A}$ and $d\alpha_{\rm B}.$ 
By the inverse function theorem,  we obtain 
\begin{lemma}\label{lem:derivatives of alpha}
  \begin{align}
     \left(\frac{\partial \alpha_{\rm A}}{\partial p} \right)_{\!T}
     &= -\frac{\alpha(1 + \alpha)q_{\rm A}}{p\left[\alpha(1 + \alpha) + q\right]}, \hspace{1ex}
\left(\frac{\partial \alpha_{\rm B}}{\partial p} \right)_{\!T}
= - \frac{\alpha(1 + \alpha)q_{\rm B}}{p\left[\alpha(1 + \alpha) + q\right]}\label{eq:alpha_Ap}\\
  \left(\frac{\partial \alpha_{\rm A}}{\partial T} \right)_{\! p}
  &= \frac{\alpha(1 + \alpha)q_{\rm A}}{T\left[\alpha(1 + \alpha) + q\right]}\left(\frac{5}{2} + \frac{T_{\rm A}}{T}\right)
  + \frac{(1 - \beta)q_{\rm A}q_{\rm B}(T_{\rm A} - T_{\rm B})}{T^2\left[\alpha(1 + \alpha) + q\right]}\label{eq:alpha_AT}\\
\left(\frac{\partial \alpha_{\rm B}}{\partial T} \right)_{\! p}
&= \frac{\alpha(1 + \alpha)q_{\rm B}}{T\left[\alpha(1 + \alpha) + q\right]}\left(\frac{5}{2} + \frac{T_{\rm B}}{T}\right)
  + \frac{\beta q_{\rm A}q_{\rm B}(T_{\rm B} - T_{\rm A})}{T^2\left[\alpha(1 + \alpha) + q\right]} \label{eq:alpha_BT}
  \end{align}
 \end{lemma}
We deduce from this lemma
\begin{align*}
  \left(\frac{\partial \alpha_{\rm A}}{\partial T} \right)_{\! p}
  &= - \frac{p}{T}\left(\frac{5}{2} + \frac{T_{\rm A}}{T}\right)\left(\frac{\partial \alpha_{\rm A}}{\partial p} \right)_{\!T}
   + \frac{(1 - \beta)q_{\rm A}q_{\rm B}(T_{\rm A} - T_{\rm B})}{T^2\left[\alpha(1 + \alpha) + q\right]}\\
\left(\frac{\partial \alpha_{\rm B}}{\partial T} \right)_{\! p}
&=  - \frac{p}{T}\left(\frac{5}{2} + \frac{T_{\rm B}}{T}\right)\left(\frac{\partial \alpha_{\rm B}}{\partial p} \right)_{\!T}
  + \frac{\beta q_{\rm A}q_{\rm B}(T_{\rm B} - T_{\rm A})}{T^2\left[\alpha(1 + \alpha) + q\right]}.
\end{align*}
Thus we obtain useful lemmas:
\begin{lemma}\label{eq:d baralpha/dT}
$$
  -\frac{T}{p}\left(\frac{\partial \alpha}{\partial T} \right)_{\! p}
%
  =  \frac{5}{2}\left(\frac{\partial \alpha}{\partial p} \right)_{\!T} + \frac{\beta T_{\rm A}}{T}\left(\frac{\partial \alpha_{\rm A}}{\partial p} \right)_{\!T} +  \frac{(1 - \beta)T_{\rm B}}{T}\left(\frac{\partial \alpha_{\rm B}}{\partial p} \right)_{\!T}
$$
\end{lemma}
\begin{lemma}\label{lem:eta_p}
\begin{equation}\label{eq:eta_p}
  \left(\frac{\partial \eta}{\partial p} \right)_{\! T}
  = - \frac{1 + \alpha}{p} 
+ \beta \left(\frac{5}{2} + \frac{T_{\rm A}}{T}\right)\left(\frac{\partial \alpha_{\rm A}}{\partial p} \right)_{\!T} + (1 - \beta)\left(\frac{5}{2} + \frac{T_{\rm B}}{T}\right)\left(\frac{\partial \alpha_{\rm B}}{\partial p} \right)_{\!T}
\end{equation}
\end{lemma}

\section{Construction of Entropy Function}\label{sec:entropy construction}
We will construct the physical entropy function for the present model system. 
First we prove:
\begin{lemma}\label{prop:entropy mod T}
The dimensionless entropy $\eta = \frac{M}{R}S$ takes the form
\begin{align}
  \eta(p,T) & =  \log \alpha \nonumber\\
  &+ \beta \log \alpha_{\rm A} + (1 - \beta) \log \alpha_{\rm B}
 - 2 \beta \log (1 - \alpha_{\rm A})  - 2 (1 - \beta) \log (1 - \alpha_{\rm B})\nonumber\\
 &+ \beta \left(\frac{5}{2} + \frac{T_{\rm A}}{T}\right)\alpha_{\rm A} 
  + (1 - \beta)\left(\frac{5}{2} + \frac{T_{\rm B}}{T}\right) \alpha_{\rm B} + \mathcal{H}(T). \label{eq:eta + H}
\end{align}
where $\mathcal{H}$ is an arbitrary  function of $T.$
\end{lemma}
\begin{proof}
Integrating \eqref{eq:eta_p} with respect to $p,$ we have
$$
   \eta(p,T)
  = - \int \frac{1 + \alpha}{p}\,dp  
+ \beta \left(\frac{5}{2} + \frac{T_{\rm A}}{T}\right)\alpha_{\rm A} + (1 - \beta)\left(\frac{5}{2} + \frac{T_{\rm B}}{T}\right) \alpha_{\rm B}.
$$
We notice that: if $dT=0$, then it follows from \eqref{eq:pfaff_1} and Saha's equation \eqref{eq:saha bar(alpha)}
\begin{equation}\label{eq:-(1 + alpha)dp/p}
- \frac{1 + \alpha}{p} dp
   =  \frac{1 + \alpha}{q_{\rm A}} d\alpha_{\rm A} + \frac{1}{\alpha} d\alpha
   = \frac{1 + \beta \alpha_{\rm A}}{q_{\rm A}} d\alpha_{\rm A} + \frac{(1 - \beta)\alpha_{\rm B}}{q_{\rm A}} d\alpha_{\rm A}+ \frac{1}{\alpha} d\alpha.
\end{equation}
It follows from the compatibility condition \eqref{eq:compatibility}
$$
   \log \left(\frac{1}{\alpha_{\rm A}} - 1\right)  -\frac{T_{\rm A}}{T} - \log \mu_{\rm A}
    = \log \left(\frac{1}{\alpha_{\rm B}} - 1\right) - \frac{T_{\rm B}}{T}- \log\mu_{\rm B}.
$$
Hence
$$
   \frac{d \alpha_{\rm A}}{q_{\rm A}} - \frac{T_{\rm A}}{T^2} \,dT
=   \frac{d \alpha_{\rm B}}{q_{\rm B}} - \frac{T_{\rm B}}{T^2} \,dT.
$$
If $dT = 0,$  then $\frac{d \alpha_{\rm A}}{q_{\rm A}} =   \frac{d \alpha_{\rm B}}{q_{\rm B}} $ and \eqref{eq:-(1 + alpha)dp/p} is found  to be
$$
   - \frac{1 + \alpha}{p} dp 
= \frac{1 + \beta \alpha_{\rm A}}{q_{\rm A}} d\alpha_{\rm A} + \frac{(1 - \beta)\alpha_{\rm B}}{q_{\rm B}} d\alpha_{\rm B}+ \frac{1}{\alpha} d\alpha.
$$
By integrating the above expression
$$
- \int \frac{1 + \alpha}{p}\,dp = \log \alpha_{\rm A} - (1 + \beta)\log (1 - \alpha_{\rm A}) 
- (1 - \beta)\log (1 - \alpha_{\rm B}) + \log \alpha + \mathcal{H}(T)
$$
\par
In a similar manner
\begin{align*}
   - \frac{1 + \alpha}{p} dp 
&= \frac{1 + (1 -\beta) \alpha_{\rm B}}{q_{\rm B}} d\alpha_{\rm B} + \frac{\beta \alpha_{\rm A}}{q_{\rm B}} d\alpha_{\rm B}+ \frac{1}{\alpha} d\alpha\\
&= \frac{1 + (1 -\beta) \alpha_{\rm B}}{q_{\rm B}} d\alpha_{\rm B} + \frac{\beta \alpha_{\rm A}}{q_{\rm A}} d\alpha_{\rm A}+ \frac{1}{\alpha} d\alpha
\end{align*}
and hence
$$
 - \int \frac{1 + \alpha}{p}\,dp = \log \alpha_{\rm B} - (2 -  \beta)\log (1 - \alpha_{\rm B}) 
- \beta\log (1 - \alpha_{\rm A}) + \log \alpha+ \mathcal{H}(T).
$$
For symmetry, we have \eqref{eq:eta + H}.
\end{proof}
\par
Next we will determine the form of $\mathcal{H}(T),$ and then obtain the entropy function up to constant.
\begin{theorem}\label{thm:entropy}
  The dimensionless entropy function $\eta(p,T)$ takes the form
      \begin{align*}
        &\log \left[\beta \alpha_{\rm A} + (1 - \beta)\alpha_{\rm B}\right] \\
        &+ \beta \left[\log \alpha_{\rm A} - 2 \log (1 - \alpha_{\rm A})  + \frac{T_{\rm A}}{T}\right]
      + (1 - \beta) \left[\log \alpha_{\rm B} - 2 \log (1 - \alpha_{\rm B})  + \frac{T_{\rm B}}{T}\right]\\
&+ \beta \left(\frac{5}{2} + \frac{T_{\rm A}}{T}\right)\alpha_{\rm A} + (1 - \beta)\left(\frac{5}{2} + \frac{T_{\rm B}}{T}\right) \alpha_{\rm B}
 + \mathrm{const.}    
    \end{align*}
   \end{theorem}
\begin{proof}
Differentiating $\eta$ with respect to $T,$ we have
\begin{multline*}
  \left(\frac{\partial \eta}{\partial T} \right)_{\! p} 
=  \frac{1}{\alpha}\left(\frac{\partial \alpha}{\partial T} \right)_{\! p} - \frac{\beta}{\alpha_{\rm A}}\left(\frac{\partial \alpha_{\rm A}}{\partial T} \right)_{\! p} - \frac{1 - \beta}{\alpha_{\rm B}}\left(\frac{\partial \alpha_{\rm B}}{\partial T} \right)_{\! p}\\
 + \frac{2\beta}{q_{\rm A}}\left(\frac{\partial \alpha_{\rm A}}{\partial T} \right)_{\! p} + \frac{2(1 - \beta)}{q_{\rm B}}\left(\frac{\partial \alpha_{\rm B}}{\partial T} \right)_{\! p}
 - \frac{\beta  \alpha_{\rm A}T_{\rm A}}{T^2} - \frac{(1 - \beta) \alpha_{\rm B}T_{\rm B}}{T^2}\\
 + \beta \left(\frac{5}{2} + \frac{T_{\rm A}}{T}\right)\left(\frac{\partial \alpha_{\rm A}}{\partial T} \right)_{\! p}
 + (1 - \beta)\left(\frac{5}{2} + \frac{T_{\rm B}}{T}\right) \left(\frac{\partial \alpha_{\rm B}}{\partial T} \right)_{\! p} + \mathcal{H}'(T).
\end{multline*}
Using the formulas in Lemma \ref{lem:derivatives of alpha} and \ref{eq:d baralpha/dT} and  setting
$  \Sigma = \alpha(1 + \alpha) + q ,$ we find that
\begin{align*}
 & \frac{1}{\alpha}\left(\frac{\partial \alpha}{\partial T} \right)_{\! p}
  - \left[\frac{\beta}{\alpha_{\rm A}}\left(\frac{\partial \alpha_{\rm A}}{\partial T} \right)_{\! p}  \frac{1 - \beta}{\alpha_{\rm B}}\left(\frac{\partial \alpha_{\rm B}}{\partial T} \right)_{\! p}\right]
  +  \frac{2\beta}{q_{\rm A}}\left(\frac{\partial \alpha_{\rm A}}{\partial T} \right)_{\! p} + \frac{2(1 - \beta)}{q_{\rm B}}\left(\frac{\partial \alpha_{\rm B}}{\partial T} \right)_{\! p}\\
&= \frac{(1 + \alpha)\beta(1 - \beta)}{\Sigma T}\left[(1 - \alpha_{\rm A})\left(\frac{5}{2} + \frac{T_{\rm A}}{T}\right) - (1 - \alpha_{\rm B})\left(\frac{5}{2} + \frac{T_{\rm B}}{T}\right)\right](\alpha_{\rm A} - \alpha_{\rm B})\\
&  + \frac{\beta(1 - \beta)(1 + \alpha)(1 - \alpha_{\rm A})(1 - \alpha_{\rm B})(\alpha_{\rm A} - \alpha_{\rm B})(T_{\rm A} - T_{\rm B})}{\Sigma T^2} \\
&  + \frac{2\alpha(1 + \alpha)}{\Sigma T}\left[\frac{5}{2} + \frac{\beta T_{\rm A}}{T} + \frac{(1 - \beta)T_{\rm B}}{T}\right]
      - \frac{2\beta(1 - \beta)(1 + \alpha)\left[q_{\rm A} - q_{\rm B}\right](T_{\rm A} - T_{\rm B})}{\Sigma T^2}.
\end{align*}
\par
The terms involving neither $T_{\rm A}$ nor $T_{\rm B}$ are 
$$
 \frac{5}{2}\left\{\frac{(1 + \alpha)\beta(1 - \beta)}{\Sigma T}\left[(1 - \alpha_{\rm A}) - (1 - \alpha_{\rm B})\right](\alpha_{\rm A} - \alpha_{\rm B})
  + \frac{2\alpha(1 + \alpha)}{\Sigma T} \right\}
  = \frac{5(1 + \alpha)}{2T}
$$
and  the terms involving $T_{\rm A}$ and $T_{\rm B}$ are
$$
   \frac{\beta T_{\rm A}}{T^2} + \frac{(1 - \beta)T_{\rm B}}{T^2} +  \frac{\beta \alpha_{\rm A}T_{\rm A}}{T^2} + \frac{(1 - \beta)\alpha_{\rm B}T_{\rm A}}{T^2} 
$$
\par
Consequently, we have
\begin{align*}
  \left(\frac{\partial \eta}{\partial T} \right)_{\! p}
& =  \frac{5(1 + \alpha)}{2T} + \frac{\beta T_{\rm A}}{T^2} + \frac{(1 - \beta)T_{\rm B}}{T^2} \\
&  + \beta \left(\frac{5}{2} + \frac{T_{\rm A}}{T}\right)\left(\frac{\partial \alpha_{\rm A}}{\partial T} \right)_{\! p} + (1 - \beta)\left(\frac{5}{2} + \frac{T_{\rm B}}{T}\right) \left(\frac{\partial \alpha_{\rm B}}{\partial T} \right)_{\! p} + \mathcal{H}'(T)
\end{align*}
which has to be equal to \eqref{eq:d eta/dT}.
Hence
$
  \mathcal{H}'(T) = - \frac{\beta T_{\rm A}}{T^2} - \frac{(1 - \beta)T_{\rm B}}{T^2}
$
  and we obtain 
  $$
  \mathcal{H} = \frac{\beta T_{\rm A}}{T} + \frac{(1 - \beta)T_{\rm B}}{T}.
  $$
\end{proof}
 
\section{Equations of  Ionized Gas Dynamics}\label{sec:Equations Gas Dynamics}

For studying thermodynamic properties of the system \eqref{eq:system}, the Lagrangian equations \cite{Smoller} are convenient
\begin{equation}\label{eq:Lagrangian}
\left\{
  \begin{array}{l}
   v_t - u_\xi = 0,\\
   u_t + p_\xi = 0,\\
   \left(e + \frac{1}{2}u^2\right)_t + (pu)_\xi = 0
  \end{array}
\right.
\end{equation}
where $p:$ pressure, $v:$ specific volume, $e:$ specific internal energy and $u:$ flow velocity.
For $C^1$ solutions, equation $\eqref{eq:Lagrangian}_3$ can be written as
\footnote{
We notice that equation $\eta_t = 0$ is equivalent  to $\left(\rho S\right)_t + \left(\rho uS\right)_x = 0$ in Eulerian coordinates.
}
$
  \eta_{t} = 0.
$
\paragraph{Characteristic speeds and vector fields:} 
For  the set of state variables $(p,u,\eta)$ we have $v_t=v_pp_t$ and equation \eqref{eq:Lagrangian} becomes
$
   p_t-\frac{1}{v_p}u_\xi = 0,\, 
   u_t + p_{\xi} = 0,\, 
   \eta_t  = 0.
$
The {\it Lagrangian} characteristic speeds are $\lambda_{\pm} = \pm \frac{1}{\sqrt{-v_p}}$ and  $\lambda_0 = 0,$ with corresponding characteristic vectors $\boldsymbol{r}_{\pm}
= \begin{bmatrix}\pm 1\\ \sqrt{-v_p}\\ 0 \end{bmatrix}$ and $\boldsymbol{r}_0 =
 \begin{bmatrix}0\\ 0\\ 1 \end{bmatrix}$. We note that characteristic speeds and characteristic vectors are all thermodynamic quantities. 
\par
For further computation, we adopt  $(p, u, T)$ as a set of  state variables. Since $v_t - u_{\xi} = v_pp_t + v_TT_t - u_{\xi} = 0$ and $\eta_t = \eta_p p_t + \eta_T T_t = 0$, we can  write system \eqref{eq:Lagrangian} in the form
\begin{equation}\label{eq:equations p, T}
  \left\{
  \begin{array}{r}
   p_t  - \frac{\eta_T}{v_p\eta_T - v_T\eta_p}u_{\xi} = 0,\\
   u_t + p_{\xi} = 0,\\
   T_t  + \frac{\eta_p}{v_p\eta_T - v_T\eta_p}u_{\xi} = 0.
  \end{array}
\right.
\end{equation}
Characteristic speeds and vector fields are computed as the following.
\begin{lemma}\label{l:eigen}
The characteristic speeds and the corresponding characteristic vector fields of system \eqref{eq:equations p, T} are
$$
\lambda_{\pm}
=  \pm\sqrt{-\frac{\eta_T}{v_p\eta_T - v_T\eta_p}},\quad \lambda_0=0, \quad
\boldsymbol{r}_{\pm} 
= \left[
  \begin{array}{c}
    \pm 1  \\
    \frac{1}{\sqrt{-\frac{\eta_T}{v_p\eta_T - v_T\eta_p}}}\\
    \mp \frac{\eta_p}{\eta_T}
  \end{array}\right], \quad
  \boldsymbol{r}_0 = \left[
  \begin{array}{c}
    0  \\
    0\\
    1
  \end{array}\right]
$$
The eigenvalue $\lambda_0$ is linearly degenerate; a pair of Riemann invariants for $\lambda_0$  is $\{u,p \}$. A Riemann invariant for both $\lambda_\pm$ is $\eta$.
The characteristic speeds of system {\em 1-\eqref{eq:system}} are then $u + \frac{1}{\rho}\lambda_{\pm}$ and $u$.
\end{lemma}

\paragraph{Computation of $\eta_p, \eta_T, v_p, v_T$:}
Fore the sake of brevity, let us introduce the quantities: $q =  \beta q_{\rm A} + (1 - \beta) q_{\rm B}, $
\begin{align*}
  \Sigma &= \alpha(1 + \alpha)  + \beta q_{\rm A} + (1 - \beta) q_{\rm B}  = \alpha(1 + \alpha)  + q \\
  \Phi &= \beta q_{\rm A}\left(\frac{5}{2} + \frac{T_{\rm A}}{T}\right)^{\! 2} + (1 - \beta)q_{\rm B}\left(\frac{5}{2} + \frac{T_{\rm B}}{T}\right)^{\! 2},\\
  \Psi &= \beta q_{\rm A}\left(\frac{15}{4} + \frac{3T_{\rm A}}{T} + \frac{T_{\rm A}^2}{T^2}\right) + (1 - \beta) q_{\rm B}\left(\frac{15}{4} + \frac{3T_{\rm B}}{T} + \frac{T_{\rm B}^2}{T^2}\right),\\
  \Omega &=  \beta(1 - \beta)q_{\rm A}q_{\rm B}\left(\frac{T_{\rm A}}{T} - \frac{T_{\rm B}}{T}\right)^2.
  \end{align*}
Substituting \eqref{eq:alpha_Ap}, \eqref{eq:alpha_AT}, \eqref{eq:alpha_BT} into 
  \eqref{eq:d eta/dT} and 
  \eqref{eq:eta_p},  we obtain
\begin{align}
  \left(\frac{\partial \eta}{\partial p} \right)_{\! T}
  &= - \frac{1 + \alpha}{p} -  \frac{\alpha(1 + \alpha)\left[\beta\left(\frac{5}{2} + \frac{T_{\rm A}}{T}\right) q_{\rm A} + (1 - \beta)\left(\frac{5}{2} + \frac{T_{\rm B}}{T}\right)q_{\rm B}\right]}{p\Sigma}\label{eq:eta_p alpha}\\
\left(\frac{\partial \eta}{\partial T} \right)_{\! p}
  &= \frac{5}{2T}(1 + \alpha)+  \frac{\alpha(1 + \alpha)\Phi + \Omega}{T\Sigma}.\label{eq:eta_T alpha}
\end{align}
\par
Since $v = \frac{a^2}{p}T(1 + \alpha)\, (a^2 = \frac{R}{\bar{M}}), $ we have by applying Lemma \ref{lem:derivatives of alpha}
$$
a^{-2}\left(\frac{\partial v}{\partial p} \right)_{\! T}
 = -\frac{T(1 + \alpha)}{p^2}
 -    \frac{T\alpha(1 + \alpha)q}{p^2\Sigma}, \quad 
     a^{-2}\left(\frac{\partial v}{\partial T} \right)_{\! p}
     = -\left(\frac{\partial \eta}{\partial p} \right)_{\! T}.
$$
\paragraph{Computation of $\lambda_{\pm}:$}
Let us first compute $\frac{1}{a^2}(v_p\eta_T - v_T\eta_p).$
That is: $\left( \frac{1 + \alpha}{p}\right)^{\! 2}$ times of
\begin{multline*}
  -\left(1 + \frac{\alpha q}{\Sigma}\right) \left[\frac{5}{2} + \frac{\alpha \Phi }{\Sigma} + \frac{\Omega}{(1 + \alpha)\Sigma}\right]\\
  +  \left\{1 +  \frac{\alpha\left[\beta q_{\rm A}\left(\frac{5}{2} + \frac{T_{\rm A}}{T}\right)  + (1 - \beta)q_{\rm B}\left(\frac{5}{2} + \frac{T_{\rm B}}{T}\right)\right]}
  {\Sigma}\right\}^{\! 2}
  = -\left(\frac{3}{2} + \frac{\alpha \Psi + \Omega}{\Sigma} \right).
\end{multline*}
Thus we have together with \eqref{eq:eta_T alpha}
%
\begin{theorem}\label{thm:characteristic speeds}
  The characteristic speeds $\lambda_{\pm}$ take the forms $\lambda_{\pm} = \pm \lambda$ where 
  \begin{equation}\label{eq:characteristic speed}
    \lambda
=  \frac{p}{a\sqrt{T}(1 + \alpha)} \sqrt{
    \frac{\frac{5}{2}(1 + \alpha)\Sigma   + \alpha(1 + \alpha) \Phi  + \Omega}
      {\frac{3}{2}\Sigma  + \alpha\Psi + \Omega}
}.
\end{equation}
\end{theorem}
\begin{remark}[Isentropes]\label{nb:isentropes}
  In $(p,u,T)$ coordinates, an integral curve of a characteristic vector field $\boldsymbol{r}$ is a solution to the system of equations 
$
\dfrac{d}{ds}\left[\begin{array}{c}
            p\\
            u \\
            T
      \end{array}\right]
= \boldsymbol{r}
$
where $\boldsymbol{r}$ stands for $\boldsymbol{r}_{\pm}$ or $\boldsymbol{r}_0.$ For $\boldsymbol{r}_{\pm},$ we have
$$
\frac{d \eta}{ds} = \frac{\partial \eta}{\partial p} \frac{dp}{ds} + \frac{\partial \eta}{\partial T} \frac{dT}{ds}
= \pm\left(\frac{\partial \eta}{\partial p}  - \frac{\partial \eta}{\partial T} \frac{\eta_p}{\eta_T}\right) = 0
$$
and for $\boldsymbol{r}_0,$ $p =$ const. and $u =$ const. Thus,  the thermodynamic part of an integral curve is $\eta =$ const. for $1,2$-characteristic directions and $p =$ const. for $0$-characteristic field. A curve $\eta =$ const. is called an {\em isentrope\/}.
Since $\left(\frac{\partial \eta}{\partial \alpha_{\rm A}}\right)_{\! T} > 0$ (see Appendix \ref{sec:isentrope}), an isentrope is the graph of a differentiable function $\alpha_{\rm A} = \alpha_{\rm A}(T)$ defined on  $T \in (0, \infty).$ 
  \end{remark}

\section{Genuine Nonlinearity (convexity) and Inflection Loci}\label{sec:convexity}
Now, we investigate the convexity of the forward and backward  fields; each characteristic direction having the eigenvalue $\lambda_\pm$ is called {\em genuinely nonlinear} if $\boldsymbol{r}_{\pm}\nabla \lambda_{\pm} \neq 0.$ 
We have chosen characteristic vectors $\boldsymbol{r}_{\pm}$ so that 
\begin{equation}\label{eq:convex p}
   \boldsymbol{r}_{\pm}\nabla \lambda_{\pm}  = \frac{v_{pp}}{2(-v_p)^\frac32}
   = \frac{\partial \lambda}{\partial p} - \frac{\eta_p}{\eta_T}\frac{\partial \lambda}{\partial T} .
\end{equation}
Hence, genuine nonlinearity implies strict convexity (or concavity) of $v$ as a function of $p$ for fixed $S$. We refer to \cite{Menikoff-Plohr} for more insight about the failure of this condition and we will see in Remark \ref{nb:entropy condition} that the entropy increases across the shock front if $\boldsymbol{r}_{\pm}\nabla\lambda_{\pm} > 0.$  It is convenient to consider a differential operator
$$
\mathcal{R}
   = \Sigma 
  \left[\eta_T\left(\frac{\partial}{\partial p}\right)_{\!T} - \eta_p\left(\frac{\partial }{\partial T}\right)_{\! p}\right]
  $$
  which is proportional to $\boldsymbol{r}_{\pm}\nabla .$
  \par
  Computation of $\mathcal{R}\lambda$ is simple but tedious.
   First we note that Lemma \ref{lem:d eta/dp, d eta/dT},  \ref{lem:derivatives of alpha} and \ref{lem:eta_p} yield 
\begin{lemma}\label{lem:eta_T del_p - eta_p del_T :q}
\begin{align*}
\mathcal{R} \alpha_{\rm A}
&= \frac{(1 + \alpha)q_{\rm A}}{pT}\left\{\frac{\alpha(1 + \alpha)T_{\rm A}}{T} \right. \\
& \hspace{2ex}   + \left. q_{\rm B}(1 - \beta)\left[1 +  \alpha \left(\frac{5}{2} + \frac{T_{\rm B}}{T}\right)\right]\left(\frac{T_{\rm A}}{T} - \frac{T_{\rm B}}{T}\right)\right\},\\
    \mathcal{R}\alpha_{\rm B}
    &=  \frac{(1 + \alpha)q_{\rm B}}{pT}\left\{\frac{\alpha(1 + \alpha)T_{\rm B}}{T}
  + q_{\rm A}\beta\left[1 + \alpha\left(\frac{5}{2} + \frac{T_{\rm A}}{T}\right)\right]\left(\frac{T_{\rm B}}{T} - \frac{T_{\rm A}}{T}\right)\right\},\\
           \mathcal{R}\alpha
           & =  \frac{\alpha(1 + \alpha)}{pT}\left\{(1 + \alpha)\left[\beta q_{\rm A}\frac{T_{\rm A}}{T} + (1 - \beta) q_{\rm B}\frac{T_{\rm B}}{T}\right]  - \Omega\right\}.
\end{align*}
\end{lemma}
\par
The above lemma give the forms of  $\mathcal{R} q,\mathcal{R} \Sigma, \mathcal{R}\Psi$ and $\mathcal{R}\Omega.$
Employing these formulas, after a long and tedious computation, we finally find that $\mathcal{R}\log \lambda$ is the summation of the following three expressions: for brevity we denote
$
 Q_T = \beta q_{\rm A}\frac{T_{\rm A}}{T} + (1 - \beta) q_{\rm B}\frac{T_{\rm B}}{T}.
$
 $$
 \textstyle%
(1)\  \frac{1 + \alpha}{pT}\left\{2\Sigma  + \Omega  + \frac{1}{2}\alpha\left[10 q + \beta q_{\rm A}\left( \frac{7T_{\rm A}}{T} + \frac{2T_{\rm A}^2}{T^2}\right) + (1 - \beta) q_{\rm B}\left(\frac{7T_{\rm B}}{T} + \frac{2T_{\rm B}^2}{T^2}\right) \right]\right\}
$$
\begin{multline*}\textstyle
(2)\  \dfrac{1 + \alpha}{2\left[\frac{5}{2}(1 + \alpha)\Sigma + \alpha(1 + \alpha)\Phi + \Omega\right]pT}\ \text{times of}\\ \textstyle
  \alpha\left[\frac{5}{2}\Sigma + \frac{5}{2}(1 + \alpha)(1 + 2\alpha)+  (1 + 2\alpha)\Phi \right]\left[2(1 + \alpha)Q_T
  -  \Omega\right] \\ \textstyle
  + (1 - 2\alpha_{\rm A})\left\{\alpha(1 + \alpha) \frac{T_{\rm A}}{T} + (1 - \beta)q_{\rm B}\left[1 +  \alpha \left(\frac{5}{2} + \frac{T_{\rm B}}{T}\right)\right]\left(\frac{T_{\rm A}}{T} - \frac{T_{\rm B}}{T}\right) \right\}\\  \textstyle
  \times 
\left\{(1 + \alpha)\beta q_{\rm A}\left[\frac{5}{2} + \alpha  \left(\frac{5}{2} + \frac{T_{\rm A}}{T}\right)^{\! 2}\right] + \Omega\right\}
 \\ \textstyle
 + (1 - 2\alpha_{\rm B})\left\{\alpha(1 + \alpha)\frac{T_{\rm B}}{T} + 2\beta q_{\rm A}\left[1 +  \alpha \left(\frac{5}{2} + \frac{T_{\rm A}}{T}\right)\right]
\left(\frac{T_{\rm B}}{T} - \frac{T_{\rm A}}{T}\right) \right\}\\ \textstyle
\times \left\{(1 + \alpha)(1 - \beta)q_{\rm B}\left[\frac{5}{2} + \alpha  \left(\frac{5}{2} + \frac{T_{\rm B}}{T}\right)^{\! 2}\right] + \Omega\right\}
\\\textstyle
 + \frac{5}{2}\alpha(1 + \alpha)^2\left[\beta (1 - 2\alpha_{\rm A})q_{\rm A}\frac{T_{\rm A}}{T}
 + (1 - \beta) (1 - 2\alpha_{\rm B})q_{\rm B}\frac{T_{\rm B}}{T}\right]\\ \textstyle
- 2 \left\{\alpha(1 + \alpha)\left[\beta q_{\rm A}\left(\frac{5}{2} + \frac{T_{\rm A}}{T}\right)\frac{T_{\rm A}}{T} + (1 - \beta)q_{\rm B}\left(\frac{5}{2} + \frac{T_{\rm B}}{T}\right)\frac{T_{\rm B}}{T}\right] + \Omega \right\}\\ \textstyle
\times 
\left\{\alpha(1 + \alpha)  + \beta q_{\rm A}\left[1 + \alpha \left(\frac{5}{2} + \frac{T_{\rm A}}{T}\right)\right] + (1 - \beta) q_{\rm B}\left[1  +  \alpha\left(\frac{5}{2} + \frac{T_{\rm B}}{T}\right)\right]\right\}
\end{multline*}
\begin{multline*}\textstyle
(3) - \dfrac{1 + \alpha}{2\left[\frac{3}{2}\Sigma + \alpha\Psi + \Omega\right]pT}\  \text{times of}\  \textstyle
  \alpha\left[\frac{3(1 + 2\alpha)}{2} + \Psi \right] \left[2(1 + \alpha)Q_T
  -  \Omega\right]\\ \textstyle
+ (1 - 2\alpha_{\rm A})\left\{\alpha(1 + \alpha)\frac{T_{\rm A}}{T} + (1 - \beta)q_{\rm B}\left[1 +  \alpha \left(\frac{5}{2} + \frac{T_{\rm B}}{T}\right)\right]
\left(\frac{T_{\rm A}}{T} - \frac{T_{\rm B}}{T}\right)\right\}\\ \textstyle
 \times 
\left\{\beta q_{\rm A}\left[\frac{3}{2} + \alpha \left(\frac{15}{4} + \frac{3T_{\rm A}}{T} + \frac{T_{\rm A}^2}{T^2}\right)\right] + \Omega\right\}\\ \textstyle
+ (1 - 2\alpha_{\rm B})\left\{\alpha(1 + \alpha)\frac{T_{\rm B}}{T} + \beta q_{\rm A}\left[1 +  \alpha \left(\frac{5}{2} + \frac{T_{\rm A}}{T}\right)\right]
\left(\frac{T_{\rm B}}{T} - \frac{T_{\rm A}}{T}\right)\right\}\\ \textstyle
\times \left\{(1 - \beta)q_{\rm B}\left[\frac{3}{2} + \alpha \left(\frac{15}{4} + \frac{3T_{\rm B}}{T} + \frac{T_{\rm B}^2}{T^2}\right)\right] + \Omega\right\}\\ \textstyle
 + \frac{3}{2}\alpha(1 + \alpha)\left[\beta (1 - 2\alpha_{\rm A})q_{\rm A}\frac{T_{\rm A}}{T}
 + (1 - \beta) (1 - 2\alpha_{\rm B})q_{\rm B}\frac{T_{\rm B}}{T}\right]\\ \textstyle
- 2\left\{\alpha\left[\beta q_{\rm A}\left(\frac{3}{2} + \frac{T_{\rm A}}{T}\right)\frac{T_{\rm A}}{T} + (1 - \beta)q_{\rm B}\left(\frac{3}{2} + \frac{T_{\rm B}}{T}\right)\frac{T_{\rm B}}{T}\right] + \Omega\right\}\\ \textstyle
\times \left\{\alpha(1 + \alpha)  + \beta q_{\rm A}\left[1 + \alpha \left(\frac{5}{2} + \frac{T_{\rm A}}{T}\right)\right] + (1 - \beta) q_{\rm B}\left[1  +  \alpha\left(\frac{5}{2} + \frac{T_{\rm B}}{T}\right)\right]\right\}
\end{multline*}

\par
Now, we study  the {\em inflection locus} which is the point set 
\[
\mathcal{I}= \left\{(T, \alpha_{\rm A});\ \boldsymbol{r}_{\pm}\nabla \lambda_{\pm} = 0, \, T > 0,\, 0 < \alpha_{\rm A} < 1 \right\}.
\]
Since $\boldsymbol{r}_+\nabla \lambda_{+} = \boldsymbol{r}_-\nabla \lambda_{-} $, both cases lead to the same result. Obviously, $\boldsymbol{r}_{\pm}\nabla \lambda_{\pm} > 0$ for sufficiently large $T$ and we observe that $\mathcal{I}$ is located in a finite region. However it is difficult to get a sketch of $\mathcal{I}$ by purely mathematical reasoning and Fig. \ref{fig:2} shows results of numerical computations.
\begin{figure}[hbt]
\centering
 \includegraphics[width =.45\linewidth]{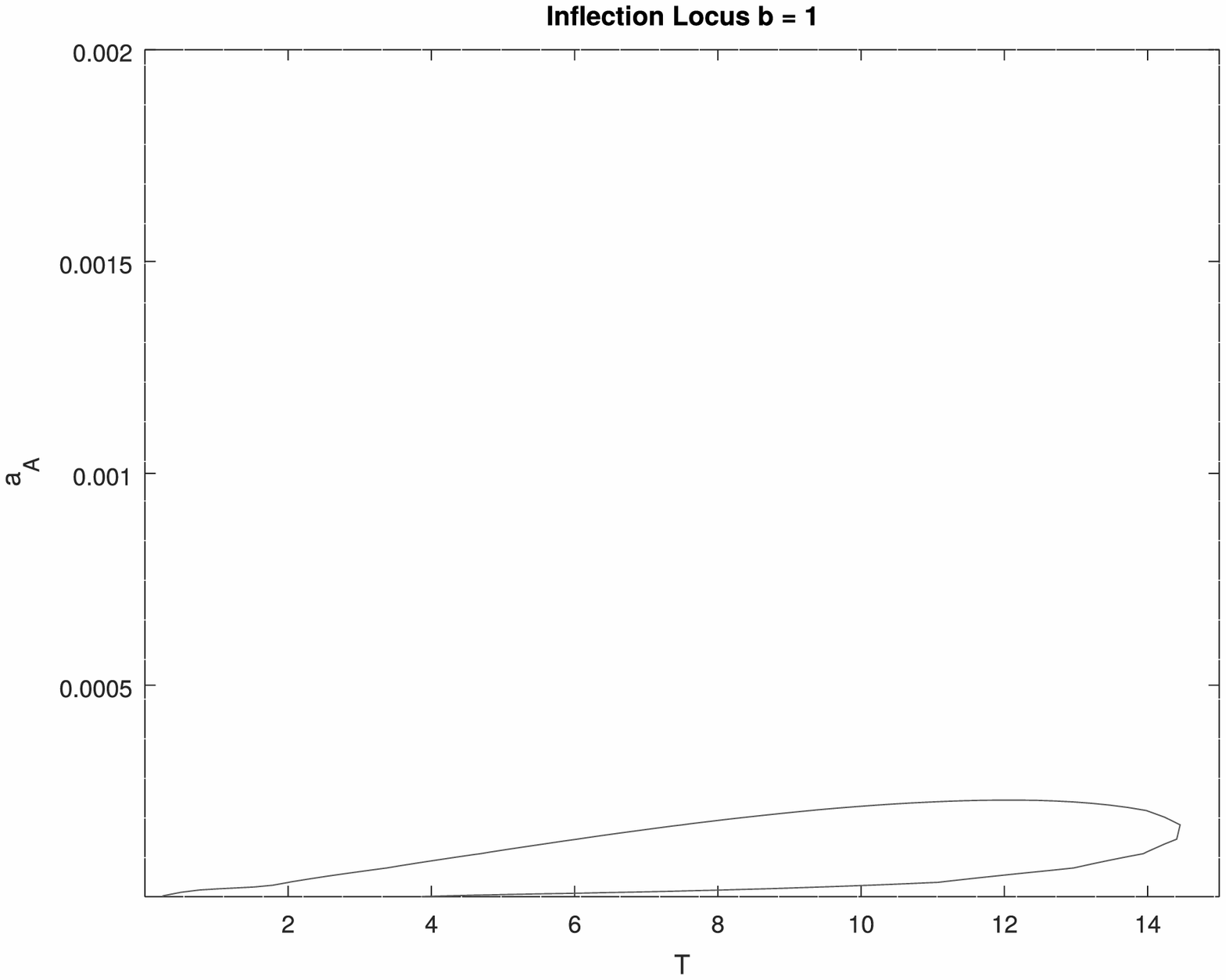}
 \includegraphics[width =.45\linewidth]{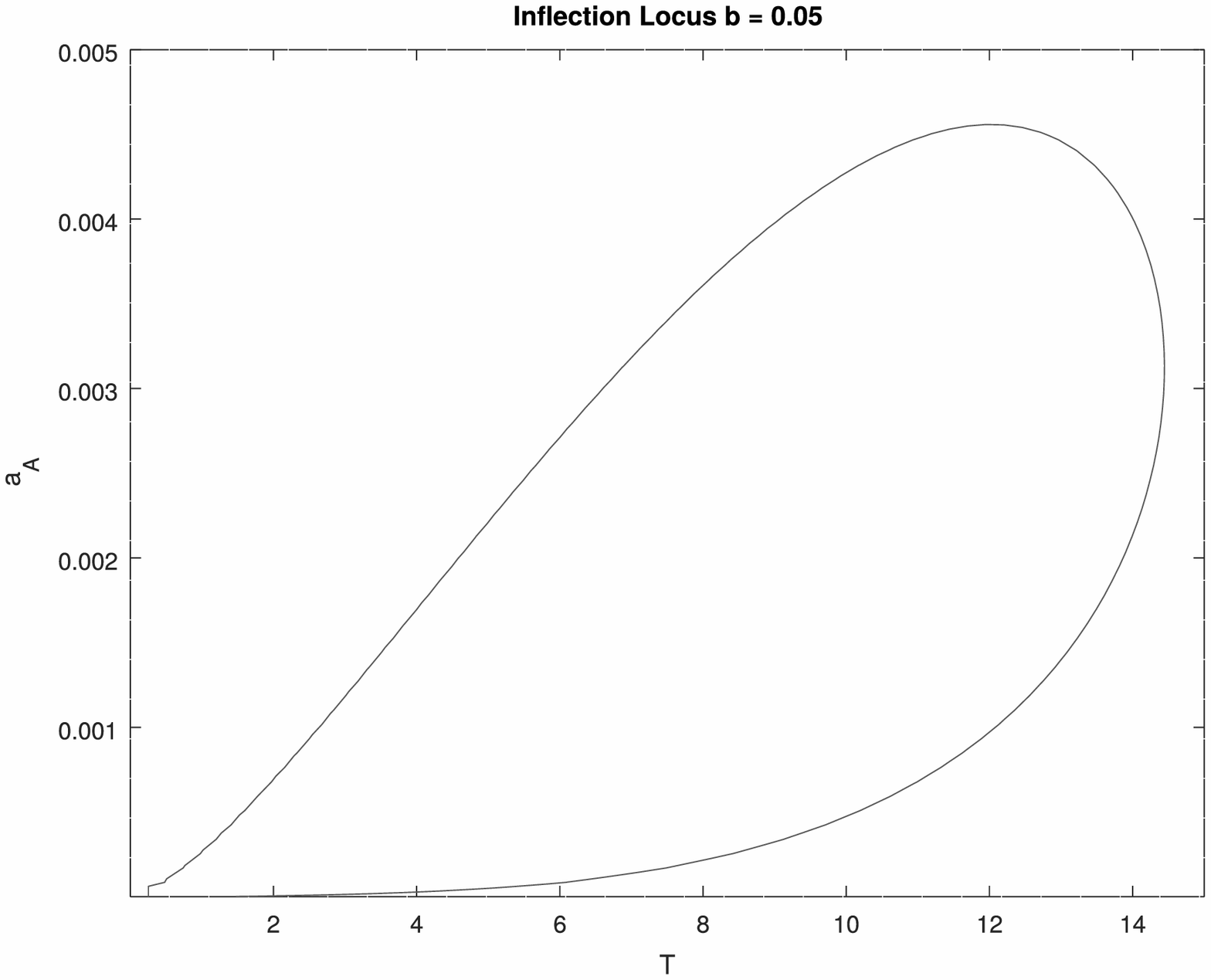}
 \caption{$T_{\rm A} = 1576.0,  T_{\rm B} = 2853.2$ left: $\beta = 1$(monatomic), right: $\beta = 0.05.$ }\label{fig:2}
\end{figure}

\par
On the other hand, it is possible to extract from the above heavy expressions asymptotics of the inflection locus for $T \to 0.$ Since $\alpha_{\rm B}$ is negligible compared with $\alpha_{\rm A},$ we observe that there are two branches such that
  $$
      \frac{\alpha_{\rm A}}{T^2} \to 0 \quad \text{or} \quad \frac{\alpha_{\rm A}}{T^2}  \to \infty.
  $$
Following theorem is a generalisation of  \cite{Asakura-Corli_ionized} Proposition 4.2.
      \begin{theorem}\label{thm:asymptotics inflection locus}
For  $T \to 0,$ the inflection locus has two branches
\begin{eqnarray*}
(1) &  \alpha_{\rm A} \sim \frac{60}{\beta} \left(\frac{T}{T_{\rm A}}\right)^{\! 3}, \quad 
	\alpha_{\rm B} \sim \frac{60\mu_{\rm A}}{\beta\mu_{\rm B}}\left(\frac{T}{T_{\rm A}}\right)^{\! 3} e^{-\frac{T_{\rm B} - T_{\rm A}}{T}}\\
(2) &   \alpha_{\rm A} \sim \frac{1}{\beta}\left(\frac{T}{T_{\rm A}}\right)^{\! \frac{3}{2}}, \quad
	\alpha_{\rm B} \sim \frac{\mu_{\rm A}}{\beta\mu_{\rm B}}\left(\frac{T}{T_{\rm A}}\right)^{\! \frac{3}{2}} e^{-\frac{T_{\rm B} - T_{\rm A}}{T}}
\end{eqnarray*}
and we conclude that the characteristic directions of $\lambda_{\pm}$ are not genuinely nonlinear in a neighbourhood of $(T, \alpha_{\rm A}) = (0,0).$
      \end{theorem}

\section{Compatibility Condition}\label{sec:compatibility}
The compatibility condition \eqref{eq:compatibility} constitutes a thermodynamic state space.
\begin{lemma}\label{lem:compatibility}
The compatibility condition takes the form
\begin{equation}\label{eq:state space}
	   \alpha_{\rm B}
    = \frac{\mu_{\rm A}\alpha_{\rm A}e^{-\frac{T_{\rm B} - T_{\rm A}}{T}}}{\mu_{\rm A}\alpha_{\rm A} e^{-\frac{T_{\rm B} - T_{\rm A}}{T}}+ \mu_{\rm B}(1 - \alpha_{\rm A})}.
\end{equation}
If $\alpha_{\rm A} \to 0,$ then $\alpha_{\rm B} \to 0$ and  we have
\begin{equation}\label{eq:compatibility small}
	\alpha_{\rm B}   = \frac{\mu_{\rm A}}{\mu_{\rm B}} \alpha_{\rm A} e^{-\frac{T_{\rm B} - T_{\rm A}}{T}}\left[1 + O(1)\alpha_{\rm A}\right].
	\end{equation}
\end{lemma}
\begin{figure}[hbt]
\centering
 \includegraphics[width =.4\linewidth]{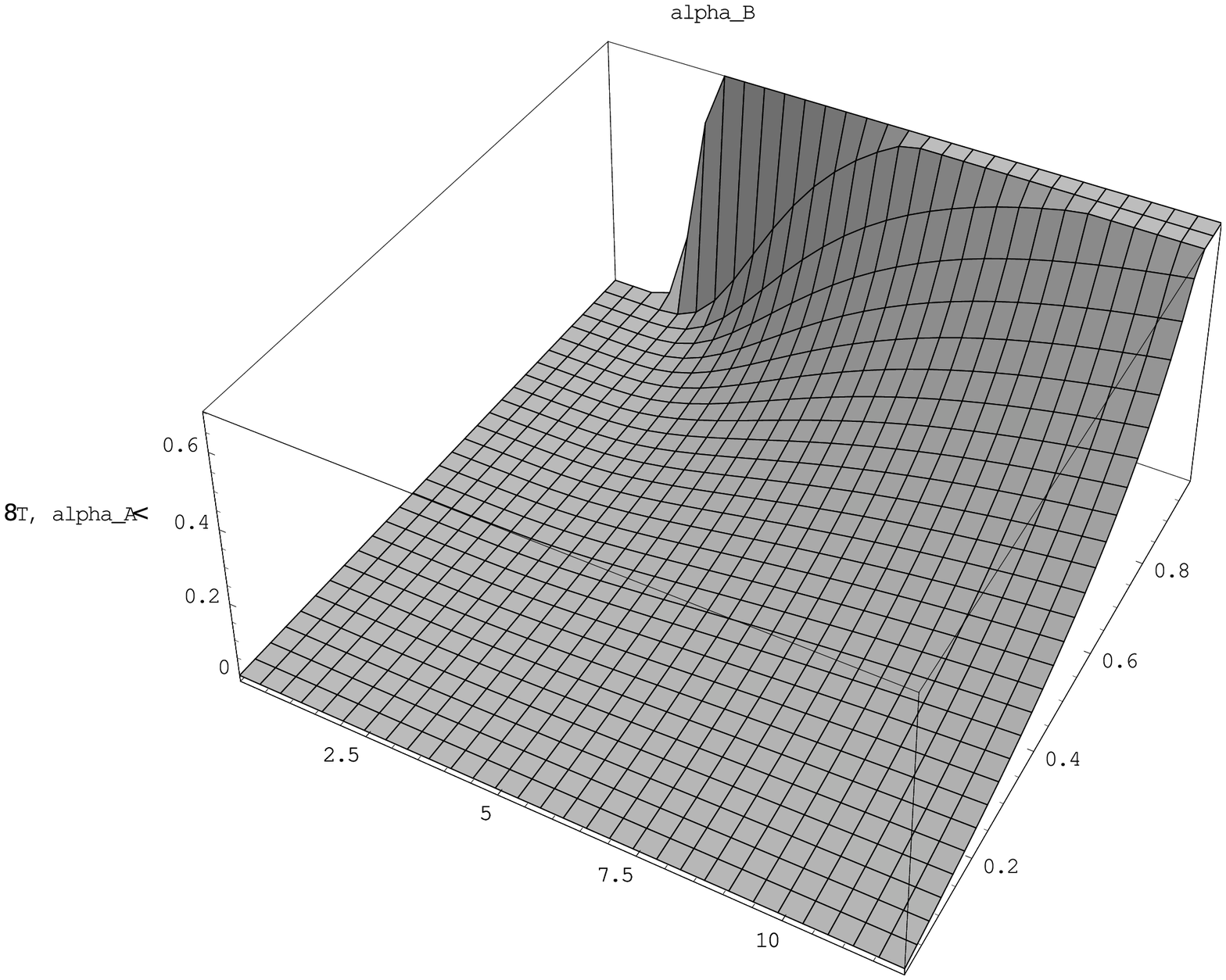}
 \hspace{4ex}
 \includegraphics[width =.4\linewidth]{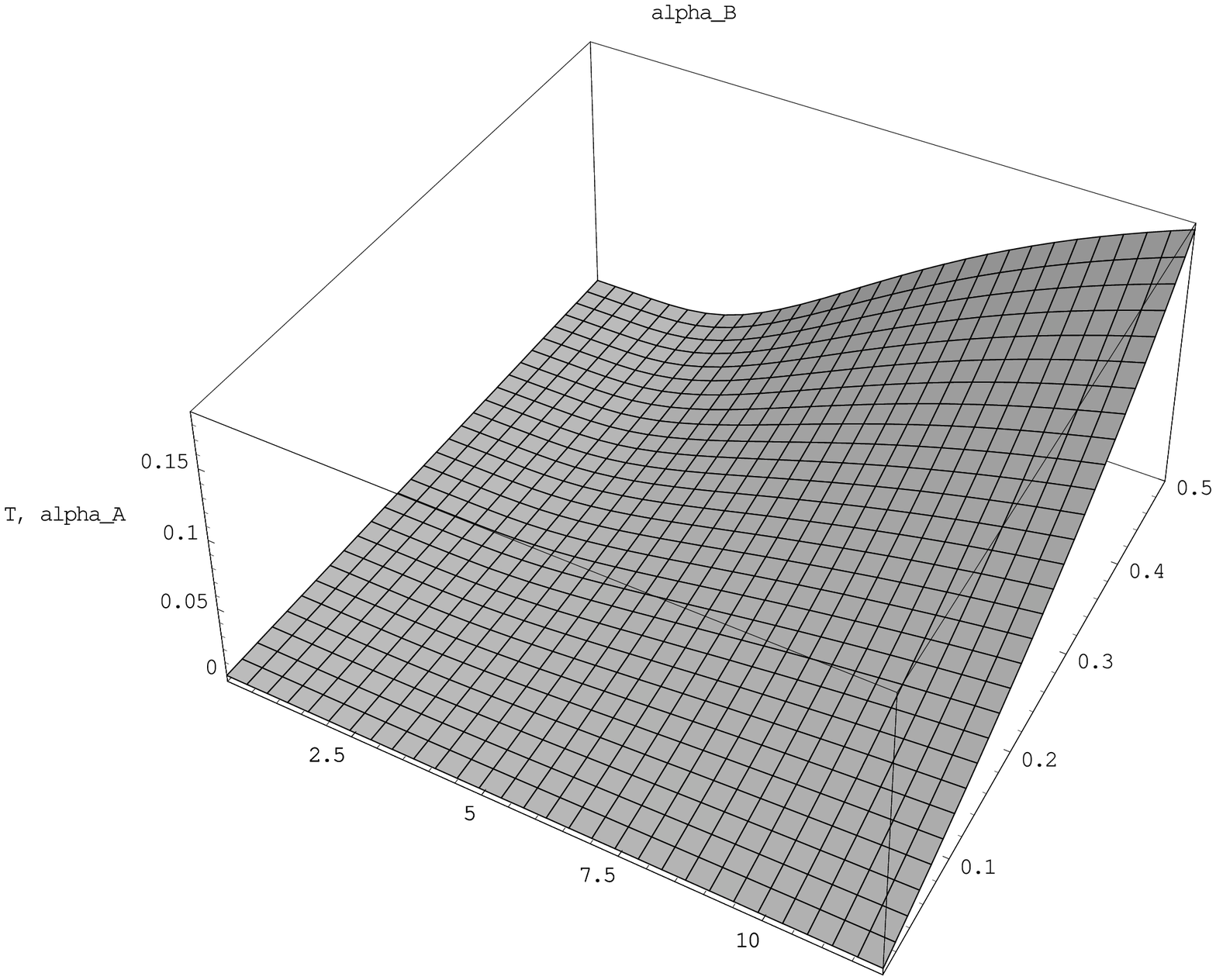}
 \caption{State space  $T_{\rm A} = 15,  T_{\rm B} = 28,$ left: $0 < \alpha_{\rm A} < 1, 0 < T < 12$, right: $0  < \alpha_{\rm A} < 0.5, 0 < T < 12$}\label{fig:1}
\end{figure}
\noindent
For A: hydrogen atom and B: helium atom, $\frac{\mu_{\rm A}}{\mu_{\rm B}} = 4. $
\par
Incidentally, we find
$$
     \alpha_{\rm B}(1 - \alpha_{\rm B})  = \frac{\mu_{\rm A}\mu_{\rm B}\alpha_{\rm A}(1 - \alpha_{\rm A})e^{-\frac{T_{\rm B} - T_{\rm A}}{T}}}{\left[\mu_{\rm A}\alpha_{\rm A} e^{-\frac{T_{\rm B} - T_{\rm A}}{T}}+ \mu_{\rm B}(1 - \alpha_{\rm A})\right]^2}
     $$

and thus derivatives of $\alpha_{\rm B}$ take the forms 
$$
	\left(\frac{\partial \alpha_{\rm B}}{\partial T}\right)_{\! \alpha_{\rm A}} 
	 = \frac{(T_{\rm B} - T_{\rm A})\alpha_{\rm B}(1 - \alpha_{\rm B})}{T^2}, \quad
	\left(\frac{\partial \alpha_{\rm B}}{\partial \alpha_{\rm A}}\right)_{\! T} 
= \frac{\alpha_{\rm B}(1 - \alpha_{\rm B})}{\alpha_{\rm A}(1 - \alpha_{\rm A})}
         $$
showing that $\left(\frac{\partial \alpha_{\rm B}}{\partial T}\right)_{\! \alpha_{\rm A}}, \left(\frac{\partial \alpha_{\rm B}}{\partial \alpha_{\rm A}}\right)_{\! T} > 0.$
By setting for brevity
$$
q = \beta q_{\rm A} + (1 - \beta)q_{\rm B},\ Q_{\rm BA} =  \frac{(1 - \beta)(T_{\rm B} - T_{\rm A})q_{\rm B}}{T},
$$
 derivatives of $\alpha$ take the  forms
\begin{lemma}\label{lem:dalpha_B/dT}
  \begin{equation}\label{eq:dalpha_B/dT}
    \left(\frac{\partial \alpha}{\partial T}\right)_{\! \alpha_{\rm A}} 
= \frac{Q_{\rm BA}}{T}, \quad
  \left(\frac{\partial \alpha}{\partial \alpha_{\rm A}}\right)_{\!T}
  = \frac{q}{q_{\rm A}}.
\end{equation}
\end{lemma}
\par
In the following sections, we shall adopt $T$ and $\alpha_{\rm A}$ as a set of independent thermodynamic state variables.

\section{Thermodynamic Hugoniot Loci}\label{sec:Hugoniot}
In the one-dimensional gas dynamics,
 the Rankine-Hugoniot conditions for a single discontinuity of constant speed $s$ are
\begin{equation}\label{eq:RH}
\left\{
\begin{array}{l}
s[\rho] = [\rho u],
\\
s\displaystyle[\rho u] = [\rho u^2 + p],
\\
s\displaystyle[\rho E] = [\rho u E + pu].
\end{array}
\right.
\end{equation}
Here we denote $[\rho]=\rho_+-\rho_-$, where $\rho_\pm$ denote the right and left limits, respectively, of $\rho$ with respect to $x$ at $x=st;$ the same notation is used for the other variables. If $[\rho]=0$ then $[u]=0$ by $\eqref{eq:RH}_1$ and $[p]=0$ by $\eqref{eq:RH}_2$; in this case, $s=u_\pm:$ the speed is equal to the flow velocity and the discontinuity is called a {\it contact discontinuity\/}. From now on we focus on the discontinuity corresponding to eigenvalues $\lambda_\pm$ and assume $[\rho]\ne0$. In this case $s$ can be eliminated  from the first equation and by substituting it into the other two equations, the conditions \eqref{eq:RH} are reduced to
\begin{equation}\label{eq:Rankine-Hugoniot}
\left\{\begin{array}{rcll}
        (u_+ - u_-)^2  + (p_+ - p_-)(v_+ - v_-) & = & 0: &\text{kinetic condition,}\\[1ex]
        e_+ -  e_- + \frac{1}{2}(p_+ + p_-)(v_+ - v_-) & = & 0: &\text{thermodynamic condition.}
                \end{array}\right.
\end{equation}
In the following, we consider a single forward shock front;  we fix a constant state $(p_+,u_+,T_+)$ and consider $(p,u,T) = (p_-,u_-,T_-)$ as a set of state variables. Under this notation, \eqref{eq:Rankine-Hugoniot} is a set of equations for {\it Hugoniot locus} of $(p_+,u_+,T_+).$
For brevity, we call solutions to $\eqref{eq:Rankine-Hugoniot}_1$ and $\eqref{eq:Rankine-Hugoniot}_2,$ respectively, the {\em kinetic} and  {\em thermodynamic Hugoniot loci. }
\par
In this section we will give a precise description of thermodynamic Hugoniot loci for the present model system and  evaluate, in particular, change of the thermodynamic variables along them; this analysis is fundamental for the study of shock waves (see \cite{Asakura-Corli_reflected}).
\par
The right thermodynamic state is denoted by $(p_+, T_+)$ and the left state $(p_-, T_-).$ 
The thermodynamic Rankine-Hugoniot condition  is written as\footnote{For the sake of convenience, we adopt the notation $\alpha_{\rm A}^{\pm}$ instead of $\alpha_{\rm A \pm}.$ }
\begin{align*}
& T_-\left(1 + \alpha^-\right)\left(4 + \frac{p_+}{p_-}\right)  + 2\left[\beta T_{\rm A}\alpha_{\rm A}^- + (1 - \beta)T_{\rm B}\alpha_{\rm B}^-\right]\\
& = T_+\left(1 + \alpha^+\right)\left(4 + \frac{p_-}{p_+}\right) + 2\left[\beta T_{\rm A}\alpha_{\rm A}^+ + (1 - \beta)T_{\rm B}\alpha_{\rm B}^+\right]
\end{align*}
The pressure is expressed as
$$
  p = \frac{(1 - \alpha_{\rm A})(1 + \alpha)}{\mu_{\rm A}\alpha_{\rm A}\alpha}T^{\frac{5}{2}} e^{-\frac{T_{\rm A}}{T}}
    = \frac{(1 - \alpha_{\rm B})(1 + \alpha)}{\mu_{\rm B}\alpha_{\rm B}\alpha} T^{\frac{5}{2}} e^{-\frac{T_{\rm B}}{T}}.
$$
and 
thus
\begin{equation}\label{eq:p_-/p_+}
\frac{p_-}{p_+}  = \frac{(1 - \alpha_{\rm A}^-)(1 + \alpha^-)\alpha_{\rm A}^+\alpha^+}
{(1 - \alpha_{\rm A}^+)( 1 + \alpha^+)\alpha_{\rm A}^- \alpha^- }
\left(\frac{T_-}{T_+}\right)^{\! \frac{5}{2}} e^{-\frac{T_{\rm A}}{T_-} + \frac{T_{\rm A}}{T_+}}, \quad
\frac{v_-}{v_+}  
= \frac{p_+T_-(1 + \alpha^-)}{p_-T_+(1 + \alpha^+)}.
\end{equation}
Consequently we have by setting $T = T_-, \alpha = \alpha^-, \alpha_{\rm A} = \alpha_{\rm A}^-$ and $\alpha_{\rm B} = \alpha_{\rm B}^-$
\begin{align}
   & \frac{T }{T_+}
    \left\{\left(1 + \alpha \right)\left[4 +  \frac{(1 - \alpha_{\rm A}^+)(1 + \alpha^+)\alpha_{\rm A} \alpha }
{(1 - \alpha_{\rm A} )( 1 + \alpha )\alpha_{\rm A}^+ \alpha^+ }
\left(\frac{T_+}{T }\right)^{\! \frac{5}{2}} e^{-\frac{T_{\rm A}}{T_+} + \frac{T_{\rm A}}{T}}\right]\right.\nonumber\\
&\hspace{6ex} \left. + 2\left[\beta \alpha_{\rm A}  \frac{T_{\rm A}}{T} + (1 - \beta)\alpha_{\rm B} \frac{T_{\rm B}}{T}\right]\rule{0ex}{4.25ex}\right\}\nonumber\\
& = 
\left(1 + \alpha^+\right)\left[4 +  \frac{(1 - \alpha_{\rm A} )(1 + \alpha )\alpha_{\rm A}^+\alpha^+}
{(1 - \alpha_{\rm A}^+)( 1 + \alpha^+)\alpha_{\rm A}  \alpha  }
\left(\frac{T}{T_+}\right)^{\! \frac{5}{2}} e^{-\frac{T_{\rm A}}{T} + \frac{T_{\rm A}}{T_+}}\right] \nonumber\\
&\hspace{6ex} + 2\left[\beta\alpha_{\rm A}^+ \frac{T_{\rm A}}{T_+} +  (1 - \beta)\alpha_{\rm B}^+\frac{T_{\rm B}}{T_+}\right]. \label{eq:Hugoniot alpha}
\end{align}
\paragraph{Asymptotics:}
We have the following asymptotic formulas.
\begin{theorem}[Asymptotics]\label{thm:Hugoniot asymptotics}
  On the thermodynamic Hugoniot locus \eqref{eq:Hugoniot alpha}, if $T \to 0,$ then $\alpha_{\rm A}, \alpha_{\rm B} \to 0$ and by setting 
$$
A = \textstyle \sqrt{\frac{\alpha_{\rm A}^+ \alpha^+ \left\{4 \left(1 + \alpha^+\right) + 2\left[\beta \alpha_{\rm A}^+ \frac{T_{\rm A}}{T_+} + (1 - \beta)\alpha_{\rm B}^+\frac{T_{\rm B}}{T_+}\right]\right\}e^{\frac{T_{\rm A}}{T_+}}}{\left[\beta + \frac{\mu_{\rm A}}{\mu_{\rm B}}(1 - \beta)\right](1 - \alpha_{\rm A}^+)(1 + \alpha^+)}},
$$
we have
  \begin{equation}\label{eq:asimptotics T(alpha) T to 0}
  \alpha_{\rm A} \sim A \left(\frac{T}{T_+}\right)^{\frac{3}{4}} e^{-\frac{T_{\rm A}}{2T}},\quad
  \alpha_{\rm B}  \sim  \frac{A\mu_{\rm A}}{\mu_{\rm B}} \left(\frac{T}{T_+}\right)^{\frac{3}{4}} e^{-\frac{2T_{\rm B} - T_{\rm A}}{2T}}.
  \end{equation}
  On the other hand, if $T \to \infty,$ then $\alpha_{\rm A}, \alpha_{\rm B} \to 1$ and
 \begin{equation}\label{eq:asymptotics T(alpha)}
   1 - \alpha_{\rm A} \sim \frac{4(1 - \alpha_{\rm A}^+)}{\alpha_{\rm A}^+\alpha^+}\left(\frac{T}{T_+}\right)^{\!-\frac{3}{2}}e^{-\frac{T_{\rm A}}{T_+}}, \  
   1 - \alpha_{\rm B}
    \sim \frac{4\mu_{\rm A}(1 - \alpha_{\rm A}^+)}{\mu_{\rm B}\alpha_{\rm A}^+\alpha^+}\left(\frac{T}{T_+}\right)^{\!-\frac{3}{2}}e^{-\frac{T_{\rm A}}{T_+}}
\end{equation}
\end{theorem}
\begin{proof}
First we let $T \to 0.$ If $\alpha_{\rm A} \geq \alpha_0 > 0,$ then the first expression of \eqref{eq:Hugoniot alpha} tends to $\infty$ and the second remains bounded, which is contradiction. Hence $\alpha_{\rm A},\, \alpha_{\rm B} \to 0.$ By 
  \eqref{eq:compatibility small}, we have 
$
   	\alpha_{\rm B}   \sim  \frac{\mu_{\rm A}}{\mu_{\rm B}} \alpha_{\rm A} e^{-\frac{T_{\rm B} - T_{\rm A}}{T}}
        $
        for $\alpha_{\rm A}, \alpha_{\rm B} \to 0$
and hence 
$
  \alpha \sim \left[\beta + \frac{\mu_{\rm A}}{\mu_{\rm B}}(1 - \beta)\right]\alpha_{\rm A}.
 $
Suppose that $\alpha_{\rm A}^2T^{-\frac{5}{2}}e^{\frac{T_{\rm A}}{T}} = O(1).$ Then the first expression tends to $0$ and the second remains bounded, which is also contradiction.
\par
We set $\alpha_{\rm A} \sim A \left(\frac{T}{T_+}\right)^{\kappa} e^{-\frac{T_{\rm A}}{2T}}$ for some $A > 0.$  Then
\begin{eqnarray*}
 & &\frac{T}{T_+}\left\{4 +  \frac{(1 - \alpha_{\rm A}^+)(1 + \alpha^+)\left[\beta + \frac{\mu_{\rm A}}{\mu_{\rm B}}(1 - \beta)\right]A^2}
{\alpha_{\rm A}^+ \alpha^+ }\left(\frac{T}{T_+}\right)^{\! 2\kappa - \frac{5}{2}} e^{-\frac{T_{\rm A}}{T_+}}\right\}\\
&\sim & \left(1 + \alpha^+\right)\left\{4 +  \frac{\alpha_{\rm A}^+\alpha^+}
{(1 - \alpha_{\rm A}^+)( 1 + \alpha^+)\left[\beta + \frac{\mu_{\rm A}}{\mu_{\rm B}}(1 - \beta)\right]A^2}
\left(\frac{T}{T_+}\right)^{\!-2 \kappa +  \frac{5}{2}} e^{\frac{T_{\rm A}}{T_+}}\right\} \\
& &  + 2\left[\beta \alpha_{\rm A}^+ \frac{T_{\rm A}}{T_+} + (1 - \beta)\alpha_{\rm B}^+\frac{T_{\rm B}}{T_+}\right].
\end{eqnarray*}
If $2\kappa - \frac{5}{2} = 0, $ then $\kappa = \frac{5}{4},$ which is impossible by the above observation. If $2\kappa - \frac{3}{2} = -2 \kappa + \frac{5}{2},$ then $\kappa = 1$ and $2\kappa - \frac{3}{2} = \frac{1}{2} > 0,$ which is also contradiction. Thus we conclude that $2\kappa - \frac{3}{2} = 0$ and hence $\kappa = \frac{3}{4},$ which implies $\alpha_{\rm A} \sim A \left(\frac{T}{T_+}\right)^{\frac{3}{4}} e^{-\frac{T_{\rm A}}{2T}},\, \alpha_{\rm B}  \sim  \frac{\mu_{\rm A}}{\mu_{\rm B}} \alpha_{\rm A} e^{-\frac{T_{\rm B} - T_{\rm A}}{T}} \sim \frac{A\mu_{\rm A}}{\mu_{\rm B}} \left(\frac{T}{T_+}\right)^{\frac{3}{4}} e^{-\frac{2T_{\rm B} - T_{\rm A}}{2T}}.$ Since $-2 \kappa +  \frac{5}{2} = 1,$ $A$ is determined by the equation
\begin{multline*}
  \frac{(1 - \alpha_{\rm A}^+)(1 + \alpha^+)\left[\beta + \frac{\mu_{\rm A}}{\mu_{\rm B}}(1 - \beta)\right]A^2 e^{-\frac{T_{\rm A}}{T_+}}}
{\alpha_{\rm A}^+ \alpha^+ }\\
 =  4 \left(1 + \alpha^+\right) + 2\left[\beta \alpha_{\rm A}^+ \frac{T_{\rm A}}{T_+} + (1 - \beta)\alpha_{\rm B}^+\frac{T_{\rm B}}{T_+}\right].
\end{multline*}
\par
Next we let $T \to \infty.$
If $\alpha_{\rm A}  \leq 1 - \delta_0\, (\delta_0 > 0),$ then the first expression of \eqref{eq:Hugoniot alpha} goes to $0$ and the second $\infty,$ which is contradiction. Thus $\alpha_{\rm A} \to 1$ as $T \to \infty$ which implies $\alpha_{\rm B} \to 1$ and hence $\alpha \to 1.$ Suppose that $(1 - \alpha_{\rm A})T^{\frac{5}{2}} = O(1).$ Then the first expression is $O(1)T$ and second $O(1),$ which is also contradiction. We may set $1 - \alpha_{\rm A} = B\left(\frac{T}{T_+}\right)^{-\kappa}.$ Then
\begin{eqnarray*}
 && \frac{2T}{T_+}\left[4 +  \frac{(1 - \alpha_{\rm A}^+)(1 + \alpha^+)}
{2B \alpha_{\rm A}^+ \alpha^+ }
\left(\frac{T}{T_+}\right)^{\! \kappa -\frac{5}{2}} e^{-\frac{T_{\rm A}}{T_+}}\right] \\
& \sim &2\left(1 + \alpha^+\right)\left[2 +  \frac{B\alpha_{\rm A}^+\alpha^+}
{(1 - \alpha_{\rm A}^+)( 1 + \alpha^+) }
\left(\frac{T}{T_+}\right)^{\! -\kappa + \frac{5}{2}} e^{\frac{T_{\rm A}}{T_+}}\right]\\
  &&
  + 2\left[\beta \alpha_{\rm A}^+ \frac{T_{\rm A}}{T_+} + (1 - \beta)\alpha_{\rm B}^+\frac{T_{\rm B}}{T_+}\right].
\end{eqnarray*}
If $\kappa -\frac{5}{2} = 0,$ then the first expression is $O(1)T$ and second $O(1),$ which is  impossible.  If $\kappa -\frac{3}{2} = -\kappa +\frac{5}{2},$ then $\kappa = 2.$ In this case, the first expression is $O(1)T$ and second $O(1),$ which is  also impossible.  Thus we find  $-\kappa + \frac{5}{2} = 1$ and hence $\kappa = \frac{3}{2}.$
We have $4 = \frac{B\alpha_{\rm A}^+\alpha^+e^{\frac{T_{\rm A}}{T_+}}}{1 - \alpha_{\rm A}^+}$ and thus obtain the asymptotic form of $\alpha_{\rm A}.$ Formula for $\alpha_{\rm B}$ is derived from
$$
1 - \alpha_{\rm B}
    = \frac{\mu_{\rm B}(1 - \alpha_{\rm A})}{\mu_{\rm A}\alpha_{\rm A} e^{-\frac{T_{\rm B} - T_{\rm A}}{T}}+ \mu_{\rm B}(1 - \alpha_{\rm A})}
$$
\end{proof}
\paragraph{Loss of Monotonicity:}
For single monatomic gases, Hugoniot loci are graphs of strictly increasing functions in $(T, \alpha)$ plane (\cite{Asakura-Corli_ionized}, \cite{Asakura-Corli_reflected}). We will show in this subsection that it is not always the case for  mixed monatomic gases. 
Let us denote
$$
\Theta^+ = \left(\frac{T_+}{T}\right)^{\! \frac{5}{2}} e^{-\frac{T_{\rm A}}{T_+} + \frac{T_{\rm A}}{T}}, \quad
\Theta_+ = \left(\frac{T}{T_+}\right)^{\! \frac{5}{2}} e^{-\frac{T_{\rm A}}{T} + \frac{T_{\rm A}}{T_+}}
$$
and
$$
   K^+ =\frac{(1 - \alpha_{\rm A}^+)(1 + \alpha^+)\alpha_{\rm A}\alpha}
{(1 - \alpha_{\rm A})( 1 + \alpha)\alpha_{\rm A}^+ \alpha^+ }, \quad
   K_+ =\frac{(1 - \alpha_{\rm A})(1 + \alpha)\alpha_{\rm A}^+\alpha^+}
{(1 - \alpha_{\rm A}^+)( 1 + \alpha^+)\alpha_{\rm A} \alpha }.
$$
Note that $K^+ \to 0,\, K_+ \to \infty$ as $\alpha_{\rm A} \to 0$ and $K^ +\to \infty,\, K_+ \to 0$ as $\alpha_{\rm A} \to 1.$
Obviously
$$
  \frac{p_+}{p} = K^+ \Theta^+, \ 
 \frac{p}{p_+} = K_+ \Theta_+\ 
 $$
  and 
\begin{equation}\label{eq:d Theta/dT}
\frac{d\Theta^+}{dT} = -\frac{1}{T}\left(\frac{5}{2} + \frac{T_{\rm A}}{T}\right)\Theta^+, \quad
\frac{d\Theta_+}{dT} = \frac{1}{T}\left(\frac{5}{2} + \frac{T_{\rm A}}{T}\right)\Theta_+.
\end{equation}
It follows from 
 \eqref{eq:dalpha_B/dT} that
\begin{align}
& \left(\frac{\partial K^+}{\partial \alpha_{\rm A}}\right)_{\! T} 
= \frac{K^+}{q_{\rm A}}\left[1 + \frac{q}{\alpha(1 + \alpha)}\right], 
&&  \left(\frac{\partial K_+}{\partial \alpha_{\rm A}}\right)_{\! T} 
= -\frac{K_+}{q_{\rm A}}\left[\frac{1}{q_{\rm A}} + \frac{q}{\alpha(1 + \alpha)}\right]\label{eq:dOmega/dalpha}\\
&  \left(\frac{\partial K^+}{\partial T}\right)_{\! \alpha_{\rm A}} 
= \frac{K^+Q_{\rm BA}}{\alpha(1 + \alpha)T}, 
&&  \left(\frac{\partial  K_+}{\partial T}\right)_{\! \alpha_{\rm A}} 
= -\frac{K_+Q_{\rm BA}}{\alpha(1 + \alpha)T^2}.\label{eq:dOmega/dT}
\end{align}
\par
By defining 
\begin{align}
 H (T, \alpha_{\rm A}) & =  T\left(1 + \alpha\right)\left(4 +  K^+\Theta^+\right) + 2\left[\beta\alpha_{\rm A} T_{\rm A} +  (1 - \beta)\alpha_{\rm B}T_{\rm B}\right]\nonumber\\
& - T_+\left(1 + \alpha^+\right)\left(4 +  K_+\Theta_+\right)  - 2\left[\beta \alpha_{\rm A}^+ T_{\rm A} + (1 - \beta)\alpha_{\rm B}^+T_{\rm B}\right],\label{eq:H}
\end{align}
the Rankine-Hugoniot condition \eqref{eq:Hugoniot alpha} is equivalent to $H (T, \alpha_{\rm A}) = 0.$
\par
Using  \eqref{eq:dalpha_B/dT} and \eqref{eq:dOmega/dalpha}, 
we have
\begin{align}
  \left(\frac{\partial H}{\partial \alpha_{\rm A}}\right)_{\! T} 
  & =  \frac{T\left(1 + \alpha\right)}{q_{\rm A}}\left[1 + \frac{q}{\alpha(1 + \alpha)}\right]\frac{p_+}{p}
  + \frac{T_+\left(1 + \alpha^+\right)}{q_{\rm A}}\left[1 + \frac{q}{\alpha(1 + \alpha)}\right]\frac{p}{p_+}\nonumber\\
&+ \frac{Tq}{q_{\rm A}} \left(4 +  \frac{p_+}{p}\right)
 + \frac{2}{q_{\rm A}}\left[\beta q_{\rm A}T_{\rm A} +  (1 - \beta)q_{\rm B}T_{\rm B}\right] , \label{eq:dH/d alpha_2}
\end{align}
showing that $\left(\frac{\partial H}{\partial \alpha_{\rm A}}\right)_{\! T}  > 0.$ In the same way
\begin{align}
  & \left(\frac{\partial H}{\partial T}\right)_{\! \alpha_{\rm A}}
  = 4(1 + \alpha)\left[1 + \frac{Q_{\rm BA}}{1 + \alpha}\left(1 + \frac{T_{\rm B}}{2T}\right)\right]\nonumber\\
&- (1 + \alpha)\left(\frac{3}{2}+ \frac{T_{\rm A}}{T}
- \frac{Q_{\rm BA}}{\alpha}\right)\frac{p_+}{p}
 - (1 + \alpha^+)\frac{T_+}{T}\left[\frac{5}{2} + \frac{T_{\rm A}}{T}
 - \frac{Q_{\rm BA}}{\alpha(1 + \alpha)}\right]\frac{p}{p_+}.\label{eq:dH/dT}
\end{align}

\begin{theorem}\label{thm:Hugoniot graph}
For every $T>0,$ there is a unique $0 < \alpha_{\rm A} < 1$ such that $H (T, \alpha_{\rm A}) =  0$ and the function $\alpha_{\rm A} =  \alpha_{\rm A}(T)$ is differentiable. 
\end{theorem}
\begin{proof}
For every fixed $T>0,$ $H (T, \alpha_{\rm A}) \to -\infty$ as $\alpha_{\rm A} \to 0,$ and $H (T, \alpha_{\rm A}) \to \infty$ as $\alpha_{\rm A} \to 1.$ Thus there is at least one $\alpha_{\rm A}$ such that $H (T, \alpha_{\rm A}) =  0.$ Since $\left(\frac{\partial H}{\partial \alpha_{\rm A}}\right)_{\! T}  > 0,$ such $\alpha_{\rm A}$ is uniquely determined and the correspondence $T \to \alpha_{\rm A}$ is differentiable. Thus the theorem follows.
\end{proof}
\par
Let us study the sign of $\left(\frac{\partial H}{\partial T}\right)_{\! \alpha_{\rm A}}$ at $(T_{+}, \alpha_{\rm A}^{+}).$
\begin{theorem}\label{thm:sign dH/dT}
  If $\beta$ is sufficiently close to $0,$  then
  $$
     \frac{d\alpha_{\rm A}}{dT}(T_+) = - \frac{\left(\frac{\partial H}{\partial \alpha_{\rm A}}\right)_{\! T}(T_{+}, \alpha_{\rm A}^{+})}{\left(\frac{\partial H}{\partial T}\right)_{\alpha_{\rm A}}(T_{+}, \alpha_{\rm A}^{+}) } < 0,
     $$
 showing that  $\alpha_{\rm A}$ is a decreasing function of $T$ in a neighbourhood of $T = T_+,$
\end{theorem}
\begin{proof}
We find by the above expression that $\left(\frac{\partial H}{\partial T}\right)_{\! \alpha_{\rm A}} > 0$ if and only if the following $F(T, \alpha_{\rm A}) $ is negative.
\begin{align*}
 F(T, \alpha_{\rm A}) & = \frac{1}{4}\left(\frac{3}{2}+ \frac{T_{\rm A}}{T}
-  \frac{Q_{\rm BA}}{\alpha}\right)\frac{p_+}{p}\\
& +\frac{T_+(1 + \alpha^+)}{4T(1 + \alpha)}\left[\frac{5}{2} + \frac{T_{\rm A}}{T}
 -\frac{Q_{\rm BA}}{\alpha(1 + \alpha)}\right]\frac{p}{p_+}
  -\left[1 + \frac{Q_{\rm BA}}{1 + \alpha}\left(1 + \frac{T_{\rm B}}{2T}\right)\right] 
\end{align*}

\par
 We set $T = T_+$ and  $\alpha_{\rm A} = \alpha_{\rm A}^{+}.$ Since $\frac{p_+}{p} = \frac{T_+}{T} = \frac{\alpha^+}{\alpha} = 1,$ we have
\begin{align*}
 F(T, \alpha_{\rm A})
& =1 + \frac{T_{\rm A}}{2T}
-  \frac{(2 + \alpha)Q_{\rm BA}}{4\alpha(1  +\alpha)}
-\left[1 + \frac{Q_{\rm BA}}{1 + \alpha}\left(1 + \frac{T_{\rm B}}{2T}\right)\right] \\
&= \frac{T_{\rm A}}{2T} - \left[\frac{(1 - \beta)(2 + 5\alpha)q_{\rm B}}{2\alpha(1  +\alpha)}
+ \frac{(1 - \beta)T_{\rm B}q_{\rm B}}{(1 + \alpha)T}\right]\frac{T_{\rm B} - T_{\rm A}}{2T}
\end{align*}
Let us consider the case: $\beta = 0.$ 
\begin{align*}
   F(T, \alpha_{\rm A}) 
&= 1 + \frac{T_{\rm A}}{2T}
-  \frac{(2 + \alpha_{\rm B})(T_{\rm B} - T_{\rm A})q_{\rm B}}{4\alpha_{\rm B}(1  +\alpha_{\rm B}) T}
 -\left[1 + \frac{(T_{\rm B} - T_{\rm A})q_{\rm B}}{(1 + \alpha_{\rm B})T}\left(1 + \frac{T_{\rm B}}{2T}\right)\right]\\
 &= \frac{T_{\rm A}}{2T}\left[1 - \frac{1  -\alpha_{\rm B}}{1  +\alpha_{\rm B}}\left(\frac{2 + 5\alpha_{\rm B}}{2} + \frac{T_{\rm B}\alpha_{\rm B}}{T}\right)\frac{T_{\rm B} - T_{\rm A}}{T_{\rm A}}\right].
\end{align*}
Obviously for any $\alpha_{\rm B} > 0,$ there is some $T > 0$ so that the above expression is negative. 
Recall that $\alpha_{\rm B}$ is a continuous function of $\alpha_{\rm A}$ and $T,$ satisfying $\left(\frac{\partial \alpha_{\rm B}}{\partial \alpha_{\rm A}} \right)_{\! T} > 0.$ Moreover $\alpha_{\rm B}(0,T) = 0$ and $\alpha_{\rm B}(1,T) = 1$ for any $T > 0.$
\par
Thus we find that: for any $T_+ > 0$ and $0 < \alpha_{\rm B}^+ < 1,$ there is a unique $\alpha_{\rm A}^+$ such that $\alpha_{\rm B}(T_+, \alpha_{\rm A}^+) = \alpha_{\rm B}^+$ which implies $ F(T_+, \alpha_{\rm A}^+) < 0.$  Since $F(T, \alpha_{\rm A})$ is continuous function of  $\beta,$ the theorem follows.
\end{proof}
\begin{figure}[hbt]
\centering
 \includegraphics[width =.485\linewidth]{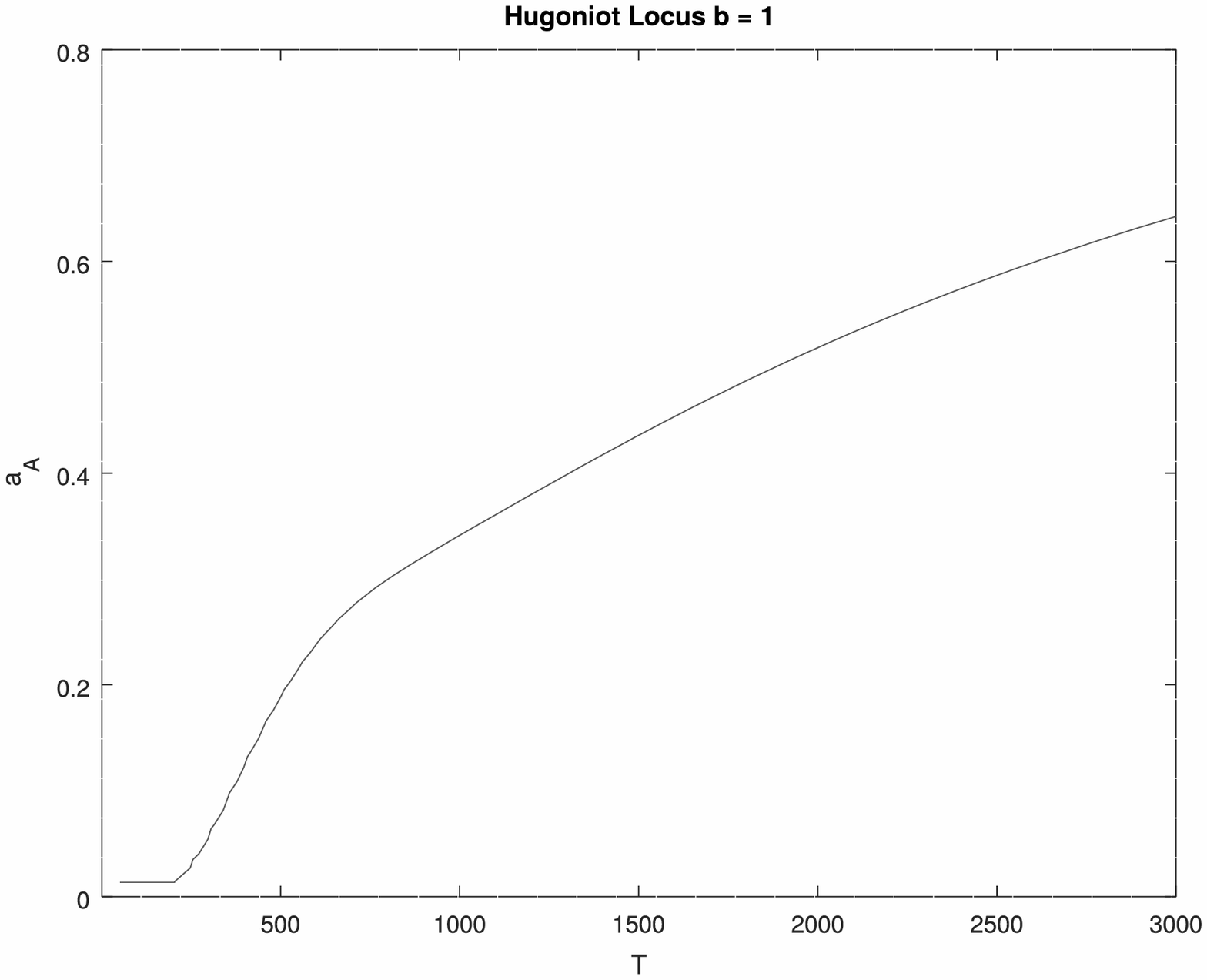}
 \includegraphics[width =.485\linewidth]{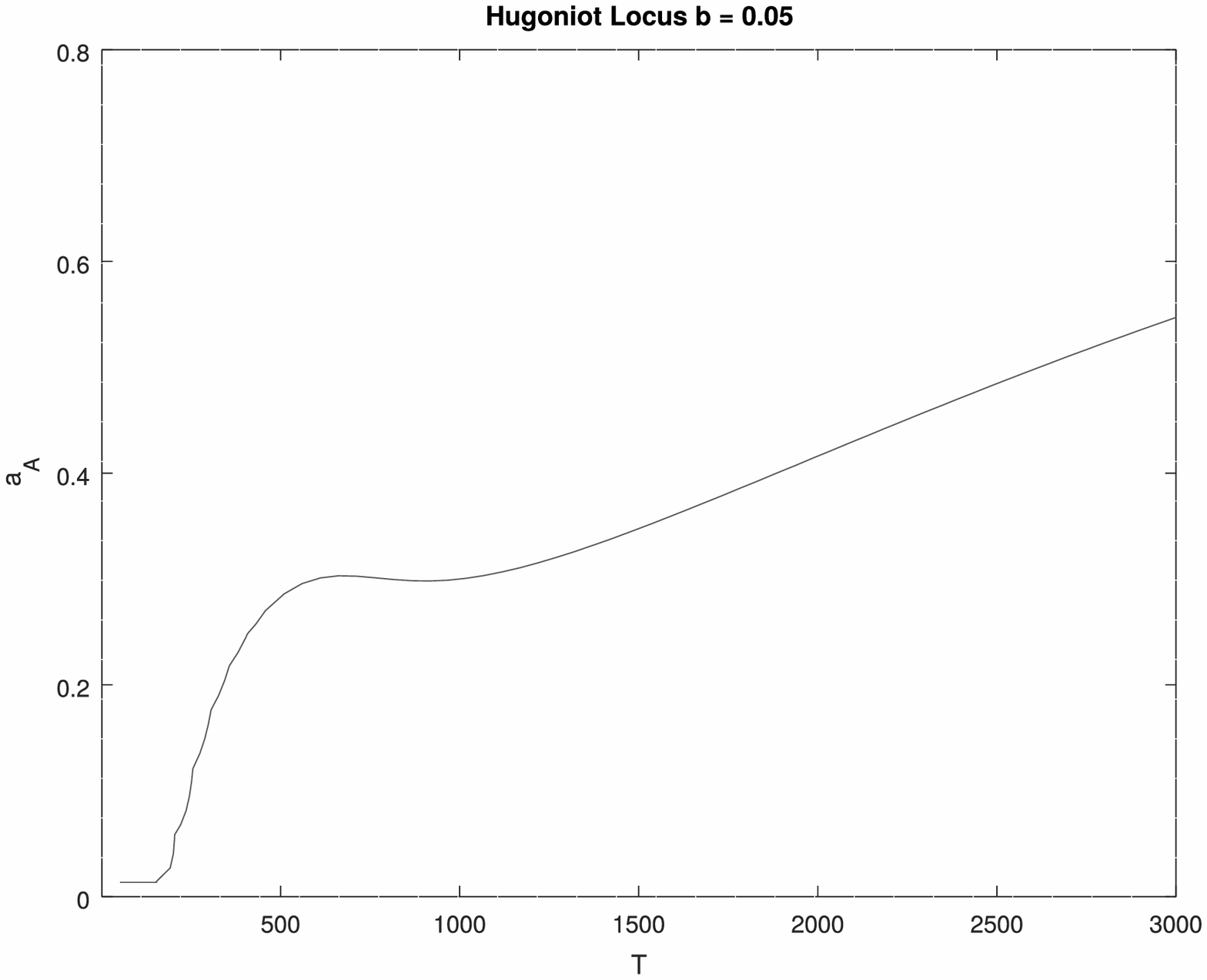}
 \caption{$T_{\rm A} = 1576.0,  T_{\rm B} = 2853.2, T= 800, \alpha_{\rm A}^+ = 0.3$ left: $\beta = 1$(single monatomic), right: $\beta = 0.05$}\label{fig:3}
\end{figure}

\paragraph{Pressure Change:}
Though the degree of ionization does not always increase across the shock front, even if the temperature increases, we will prove in this subsection:
\begin{theorem}\label{thm:dp/dT}
  The pressure $p$ strictly  increases along the Hugoniot locus as the temperature $T$ increases. 
  \end{theorem}
\begin{proof}
First we notice that:  by setting
\begin{equation}\label{eq:Q_T}
 Q_T = \frac{\beta q_{\rm A}T_{\rm A} +  (1 - \beta)q_{\rm B}T_{\rm B}}{T},\
 Q_{\rm BA} =  \frac{(1 - \beta)(T_{\rm B} - T_{\rm A})q_{\rm B}}{T}, 
\end{equation}
\eqref{eq:dH/d alpha_2} and \eqref{eq:dH/dT} together with \eqref{eq:p_-/p_+} yield
\begin{align}
   &\textstyle \left(\frac{\partial H}{\partial T}\right)_{\! \alpha_{\rm A}} 
  =   4(1 + \alpha)\left[1 + \frac{ Q_{\rm BA}}{1 + \alpha}\left(1 + \frac{T_{\rm B}}{2T}\right)\right]\nonumber \\
  &\textstyle \quad - \frac{T_+(1 + \alpha^+)}{T}\left\{ \left[\frac{5}{2} + \frac{T_{\rm A}}{T}
     + \frac{ Q_{\rm BA}}{\alpha(1 + \alpha)}\right]\frac{p}{p_+} + \left(\frac{3}{2}+ \frac{T_{\rm A}}{T}
- \frac{ Q_{\rm BA}}{\alpha}\right)\frac{v}{v_+}\right\} \label{eq:dH/dT 2}\\
&\textstyle \left(\frac{\partial H}{\partial \alpha_{\rm A}}\right)_{\! T} 
=  \frac{4}{q_{\rm A}}\left(q  + \frac{Q_T}{2}\right) +  \frac{T_+\left(1 + \alpha^+\right)}{q_{\rm A}}\left\{ \left[1 + \frac{q}{\alpha(1 + \alpha)}\right]\frac{p}{p_+}
+ \left(1 + \frac{q}{\alpha}\right)\frac{v}{v_+}\right\} \label{eq:dH/dalpha 2}
\end{align}
\par
 We now compute $\frac{dp}{dT}$ by differentiating both sides of 
$$
   \log p = \log (1 - \alpha_{\rm A}) - \log \alpha_{\rm A} + \log (1 + \alpha) - \log \alpha + \frac{5}{2}\log T - \frac{T_{\rm A}}{T} + \text{const.}
   $$
Using \eqref{eq:dalpha_B/dT} and \eqref{eq:Q_T}, we have
$$
           \frac{d\alpha}{dT} 
           = \beta \frac{d\alpha_{\rm A}}{dT} + (1 - \beta) \left(\frac{\partial\alpha_{\rm B}}{\partial T} + \frac{\partial\alpha_{\rm B}}{\partial \alpha_{\rm A}}\frac{d\alpha_{\rm A}}{d T}  \right)
           = \frac{1}{q_{\rm A}}\left(q\frac{d\alpha_{\rm A}}{dT} + \frac{Q_{\rm BA}q_{\rm A}}{T}\right).
$$
   Thus we get
\begin{align*}
  \frac{1}{p} \frac{dp}{dT}
 &= -\frac{1}{\alpha_{\rm A}(1 - \alpha_{\rm A})}\frac{d\alpha_{\rm A}}{dT} - \frac{1}{\alpha(1 + \alpha)}\left(\frac{q}{q_{\rm A}}\frac{d\alpha_{\rm A}}{dT} + \frac{Q_{\rm BA}}{T}\right) + \frac{1}{T}\left(\frac{5}{2} + \frac{T_{\rm A}}{T}\right) \\
  &= \frac{\left[1 + \frac{q}{\alpha(1 + \alpha)}\right]\left(\frac{\partial H}{\partial T}\right)_{\! \alpha_{\rm A}}  + \frac{q_{\rm A}}{T}\left[\frac{5}{2} + \frac{T_{\rm A}}{T} - \frac{Q_{\rm BA}}{\alpha(1 + \alpha)}\right]\left(\frac{\partial H}{\partial \alpha_{\rm A}}\right)_{\! T}}
       {q_{\rm A}\left(\frac{\partial H}{\partial \alpha_{\rm A}}\right)_{\! T}}.
\end{align*}
Since $\frac{q_{\rm A}}{T} \left(\frac{\partial H}{\partial \alpha_{\rm A}}\right)_{\! T} > 0,$ we examine the numerator, which is computed as
\begin{multline*}
\textstyle  \frac{T_+\left(1 + \alpha^+\right)}{T}\!\left\{\!\left(1 + \frac{q}{\alpha}\right)\!\left[\frac{5}{2} + \frac{T_{\rm A}}{T} - \frac{Q_{\rm BA}}{\alpha(1 + \alpha)}\right]
 - \left(\frac{3}{2}+ \frac{T_{\rm A}}{T}
-\frac{ Q_{\rm BA}}{\alpha}\right)\!\left[1 + \frac{q}{\alpha(1 + \alpha)}\right]\right\}\!\frac{v}{v_+}\\
\textstyle  + 4(1 + \alpha)\!\left[1 + \frac{ Q_{\rm BA}}{1 + \alpha}\left(1 + \frac{T_{\rm B}}{2T}\right)\right]\!\left[1 + \frac{q}{\alpha(1 + \alpha)}\right]+ 4\left(q  + \frac{Q_T}{2}\right)\!\left[\frac{5}{2} + \frac{T_{\rm A}}{T} - \frac{Q_{\rm BA}}{\alpha(1 + \alpha)}\right].
\end{multline*}
Note that
\begin{align*}
&   \frac{5}{2}\left(1 + \frac{q}{\alpha}\right) - \frac{3}{2}\left[1 + \frac{q}{\alpha(1 + \alpha)}\right]
= 1 + \frac{1 + \frac{5}{2}\alpha}{\alpha(1 + \alpha)},  \\
& \frac{T_{\rm A}}{T}\left[1 + \frac{q}{\alpha} - 1 - \frac{q}{\alpha(1 + \alpha)}\right] = \frac{qT_{\rm A}}{T(1 + \alpha)}
\end{align*}
and
$$
  -\left(1 \!+\! \frac{q}{\alpha}\right)\!\frac{Q_{\rm BA}}{\alpha(1 + \alpha)}
\!+\! \frac{ Q_{\rm BA}}{\alpha}\!\left[1 \!+\! \frac{q}{\alpha(1 + \alpha)}\right]
= -\frac{Q_{\rm BA}}{\alpha(1 + \alpha)} \!+\! \frac{ Q_{\rm BA}}{\alpha}
= \frac{ Q_{\rm BA}}{1 + \alpha}.
$$
Then we find that
\begin{align*}
  &\left(1 + \frac{q}{\alpha}\right)\left[\frac{5}{2} + \frac{T_{\rm A}}{T} - \frac{Q_{\rm BA}}{\alpha(1 + \alpha)}\right]
- \left(\frac{3}{2}+ \frac{T_{\rm A}}{T}  -\frac{ Q_{\rm BA}}{\alpha}\right)\left[1 + \frac{q}{\alpha(1 + \alpha)}\right]\\
&= 1 + \frac{1 + \frac{5}{2}\alpha}{\alpha(1 + \alpha)} + \frac{qT_{\rm A}}{T(1 + \alpha)} +  \frac{ Q_{\rm BA}}{1 + \alpha} > 0.
\end{align*}
Moreover
\begin{align*}
&4Q_{\rm BA}\left(1 + \frac{T_{\rm B}}{2T}\right)\left[1 + \frac{q}{\alpha(1 + \alpha)}\right]
- 4\left(q  + \frac{Q_T}{2}\right)\frac{Q_{\rm BA}}{\alpha(1 + \alpha)}\\
&> 2Q_{\rm BA}\left[\frac{T_{\rm B}}{T} - \frac{Q_T}{\alpha(1 + \alpha)}\right]  > 2Q_{\rm BA}\left[\frac{T_{\rm B}}{T} - \frac{Q_T}{\alpha}\right] 
\end{align*}
and noticing
$$
  \frac{T_{\rm B}}{T} - \frac{Q_T}{\alpha}
  = \frac{1}{T}\left[T_{\rm B} - \frac{\beta q_{\rm A}T_{\rm A} + (1 - \beta)q_{\rm B}T_{\rm B}}{\alpha}\right]
  > \frac{T_{\rm B}}{T}\left(1 - \frac{q}{\alpha}\right) \geq 0,
$$
we conclude that the numerator is strictly positive and hence the theorem is proved.
\end{proof}
\begin{remark}\label{nb:entropy condition}
 It is well known that  (\cite[\S 86]{Landau-Lifshitz_SP} ): if $(p_0, u_0, S_0)$ and $(p_1,u_1,S_1)$ are connected by a shock front, then 
\begin{equation}\label{eq:entropy variation}
   S_1 -  S_0 = \frac{1}{12T_0} \left(\frac{\partial^2 v}{\partial p_0^2}\right)_{\substack{S\\ \rule{1ex}{0ex}}}(p_1 - p_0)^3 + O(1)(p_1- p_0)^4.
\end{equation}
This formula, first obtained by H. Bethe in  \cite{Bethe}, is notable, because it depends on neither the particular equation of state nor the form of internal structure. In particular, it is true for present mixed ionized system of equations.
\par
Suppose that $v_{pp}(p, S)  > 0.$ Then the entropy increases as the pressure increases. It follows from \eqref{eq:convex p} that this condition  implies that this characteristic direction is genuinely nonlinear. Consequently, if  $|p_1 - p_0|$ is sufficiently small,  the Lax condition (see \cite{Dafermos}, \cite{Smoller})  holds in this case. Thus we can call the above discontinuity a {\it shock wave\/}  as long as $p_1 > p_0$ and  $|p_1 - p_0|$ is sufficiently small.
\end{remark} 
\par For discontinuities with arbitrary amplitude
\begin{theorem}[Bethe-Weyl] \label{thm:Bethe-Weyl}
The thermodynamic Hugoniot locus of the state $(v_0, S_0)$ intersects each isentrope at least once. Moreover, if $p_{vv} > 0$ along an isentrope, then the  locus intersects it exactly once; in this case,  $|u - \sigma| < c,$ if $v_1 < v_0.$ while the opposite inequalities hold if $v_1 > v_0.$
\end{theorem}
\noindent
Hence the Lax condition holds even for large $|p_1 - p_0|$ as long as $p_1 > p_0.$ Proof is found in \cite{Bethe}, \cite{Weyl} and \cite[(3.44)]{Menikoff-Plohr}.  We may also call this \lq\lq shock wave\rq\rq, however the physical entropy does not necessarily increase. 
\par
The following theorem in  \cite{Bethe}  guarantees increase of the physical entropy.  Let us introduce the {\it Gr\"uneisen coeficient} $\Gamma$ defined by
$$
\Gamma = -\frac{v}{T}\frac{\partial^2 e}{\partial S \partial v} = \frac{v}{T}\left(\frac{\partial p}{\partial S}\right)_{\! v} = v\left(\frac{\partial p}{\partial e}\right)_{\! v}.
$$
\begin{theorem}[Bethe] \label{thm:Bethe}
Suppose that $p_{vv} > 0$ and $\Gamma \geq -2.$ Then the thermodynamic Hugoniot locus of the state $(v_0, S_0)$ intersects each isentrope exactly  once and  $S_1 >S_0$ if $v_1 < v_0,$ while  $S_1 < S_0$ if $v_1 > v_0.$
\end{theorem}
  In our case, a set of independent thermodynamic variables are $\alpha_{\rm A}$ and $T.$ The Gr\"uneisen coefficient is expressed as
    $$
  \Gamma
= \frac{v\left(\frac{\partial p}{\partial \alpha_{\rm A}}\frac{\partial v}{\partial T} - \frac{\partial v}{\partial \alpha_{\rm A}} \frac{\partial p}{\partial T}\right)}
  {T\left(\frac{\partial S}{\partial \alpha_{\rm A}}\frac{\partial v}{\partial T} - \frac{\partial v}{\partial \alpha_{\rm A}} \frac{\partial S}{\partial T}\right)},
  $$
where $ S = \frac{R}{M}\eta.$ After simple but tedious computations like in Appendix \ref{sec:computation}, we  prove finally $\Gamma > 0.$

\section{Approximation of Thermodynamic Hugoniot Loci}\label{sec:approximate Hugoniot}
Next section, we consider a forward shock front having the right state $(p_+,T_+)$ in ordinary circumstances:  the pressure $p_+$ and the temperature $T_+$ have proper finite values and $\alpha_{\rm A}^+,\,\alpha_{\rm B}^+$ are supposed to be $0.$ 
However, we find by \eqref{eq:p_-/p_+} that $\frac{p }{p_+} \to 0$ as $\alpha_{\rm A}^+, \alpha_{\rm B}^+ \to 0,$ which is contradictory. Notice that
$$
	p = \frac{(1 - \alpha_{\rm A})(1 + \alpha)}{\mu_{\rm A}\alpha_{\rm A}\alpha}T^{\frac{5}{2}}e^{-\frac{T_{\rm A}}{T}}.
$$
Then we observe that
$$
	\alpha_{\rm A}^+\alpha^+ = \frac{(1 - \alpha_{\rm A}^+)(1 + \alpha^+)}{\mu_{\rm A}p_+}T_+^{\frac{5}{2}}e^{-\frac{T_{\rm A}}{T_+}}
	\sim \frac{T_+^{\frac{5}{2}}e^{-\frac{T_{\rm A}}{T_+}}}{\mu_{\rm A}p_+}, \quad \text{as} \quad \alpha_{\rm A}^+, \alpha_{\rm B}^+ \to 0.
 $$
 Setting $\hat{L}_+^2 =  \frac{T_+^{\frac{5}{2}}}{\mu_{\rm A}p_+},$ we obtain 
$\alpha_{\rm A}^+\alpha^+ \sim \hat{L}_+^2e^{-\frac{T_{\rm A}}{T_+}}$ and an approximate formula
 $$
 	\frac{p_+}{p }  = \frac{\alpha_{\rm A}  \alpha }{(1 - \alpha_{\rm A} )(1 + \alpha )\hat{L}_+^2}.
\left(\frac{T_+}{T }\right)^{\! \frac{5}{2}} e^{\frac{T_{\rm A}}{T }}
 $$
 By letting $\alpha_{\rm A}^+ \to 0, $ the thermodynamic Rankine-Hugoniot condition takes the form
$$
\frac{T }{T_+}\left(1 + \alpha \right)\left(4 + \frac{p_+}{p }\right)  + \frac{2 \left[\beta T_{\rm A}\alpha_{\rm A}  +  (1 - \beta)T_{\rm B}\alpha_{\rm B} \right]}{T_+}
= 
4 + \frac{p }{p_+}
$$
whose  solution is called the {\it approximate thermodynamic Hugoniot locus\/} of a laboratory state $(p_+,T_+).$
\begin{theorem}\label{thm:GN small alpha^-}
Let $\sigma$ be a positive constant satisfying $\sigma < 60.$ If $ T$ and $T_+$ are sufficiently small, so that
$$\textstyle
    \sqrt{\frac{\beta T_+T_{\rm A}^{\frac{3}{2}}}{\mu_{\rm A}p_+}} \left(\frac{T }{T_{\rm A}}\right)^{\! - \frac{9}{4}}e^{- \frac{T_{\rm A}}{2T }} \leq \sigma,
$$
then the approximate thermodynamic Hugoniot locus is located in a genuinely nonlinear region for sufficiently small $\alpha_{\rm A} , \alpha_{\rm B} .$
\end{theorem}
\begin{proof}
Let us denote $\Pi = \frac{p_+}{p }.$ Then the thermodynamic Rankine-Hugoniot condition is found to be a quadratic equation of $\Pi:$
$$
  \left(1 + \alpha \right) \frac{T }{T_+}\Pi^2 + 2\left[2\left(1 + \alpha \right)\frac{T }{T_+}
+ \frac{\beta T_{\rm A}\alpha_{\rm A}  +  (1 - \beta)T_{\rm B}\alpha_{\rm B} }{T_+} - 2\right]\Pi -1 = 0
$$
Let  $\Gamma(\Pi)$ denote the left side of the above expression. Clearly $\Gamma (0) = -1 < 0.$ Since $T_+ \leq T ,$ 
$$
2\left(1 + \alpha \right)\frac{T }{T_+} + \frac{\beta T_{\rm A}\alpha_{\rm A}  +  (1 - \beta)T_{\rm B}\alpha_{\rm B} }{T_+} - 2 \geq 2\left(\frac{T }{T_+} - 1\right) \geq 0,
$$
which implies
$$
 \Gamma\left(\frac{T_+}{T }\right) \geq (1 + \alpha )\frac{T_+}{T } +  4\left(\frac{T }{T_+} - 1\right)\frac{T_+}{T } - 1 > 3\left(1 - \frac{T_+}{T }\right) \geq 0.
$$
Thus we conclude that $0 < \Pi = \frac{p_+}{p } < \frac{T_+}{T },$ which is
$$
 \frac{\mu_{\rm A} p_+ \alpha_{\rm A}  \alpha }{(1 - \alpha_{\rm A} )(1 + \alpha )T ^{\frac{5}{2}}}e^{\frac{T_{\rm A}}{T }} < \frac{T_+}{T }.
$$
\par
Since we may assume $\alpha_{\rm A}  > \alpha_{\rm B} ,$ we find that
\begin{align*}
& (1 - \alpha_{\rm A} )(1 + \alpha )\\
& = 1 + \alpha  -  \alpha_{\rm A}  - \alpha_{\rm A}  \alpha  = 1 - (1 - \beta)(\alpha_{\rm A}  - \alpha_{\rm B} )  - \alpha_{\rm A}  \alpha  < 1.
\end{align*}
Thus
$$
 \frac{\beta \mu_{\rm A} p_+ \left(\alpha_{\rm A} \right)^2}{T^{\frac{5}{2}}}e^{\frac{T_{\rm A}}{T}} <  \frac{\mu_{\rm A} p_+ \alpha_{\rm A}  \alpha }{T ^{\frac{5}{2}}}e^{\frac{T_{\rm A}}{T }} < \frac{T_+}{T }.
$$
and we have
$$\textstyle
   0 < \alpha_{\rm A}  < \sqrt{\frac{T_+T ^{\frac{3}{2}}}{\beta\mu_{\rm A}p_+}} e^{- \frac{T_{\rm A}}{2T }}
= \sqrt{\frac{T_+T_{\rm A}^{\frac{3}{2}}}{\beta \mu_{\rm A}p_+}} \left(\frac{T }{T_{\rm A}}\right)^{\! \frac{3}{4}}e^{- \frac{T_{\rm A}}{2T }}
$$
By virtue of Theorem \ref{thm:asymptotics inflection locus}, we have the theorem.
\end{proof}
\section{Shock Tube Problem}\label{sec:shock tube}
The shock tube problem consists in finding the state $(\alpha_{\rm A}^- 
, T_- ),$ for given  $(\alpha_{\rm A}^+, T_+)$ and $u_{\pm},$ satisfying Rankine-Hugoniot conditions
\begin{equation}\label{eq:RH_2}
  \left\{
  \begin{array}{ccll}
    (u_-  - u_+)^2 &= -(p_-  - p_+)(v_-  - v_+):& \text{kinetic part}, \\ [1ex]
    h_-  - h_+ &= {\textstyle \frac{1}{2}}(v_-  + v_+)(p_-  - p_+):& \text{thermodynamic part} .
  \end{array}
   \right.
\end{equation}
The kinetic part takes a form
\begin{equation}\label{eq:RH-kinetic_2}
  (u_-  - u_+)^2 = -p_+v_+\left(\frac{p_- }{p_+} - 1\right)\left(\frac{v_- }{v_+} - 1\right).
\end{equation}
Denoting  $\alpha_{\rm A}  = \alpha_{\rm A}^-, T = T_-$ and $u_A = u_-,$ we define $G$ to be 
\begin{align*}
  G(\alpha_{\rm A}, T) &=  - \frac{(u_A - u_+)^2}{p_+v_+} -\left(\frac{p }{p_+} - 1\right)\left(\frac{v }{v_+} - 1\right)\\
  &= - \frac{M(u_A - u_+)^2}{RT_+(1 + \alpha^+)} + \frac{p}{p_+} + \frac{v}{v_+} - \frac{T(1 + \alpha)}{T_+(1 + \alpha^+)} - 1 .
\end{align*}
Then the kinetic condition \eqref{eq:RH-kinetic_2} is equivalent to $G(\alpha_{\rm A}, T) = 0.$ Substituting \eqref{eq:p_-/p_+} into the above expression, we have actually
\begin{align}
   G(\alpha_{\rm A}, T)
  &= - \frac{M(u_A - u_+)^2}{RT_+(1 + \alpha^+)} 
+ \frac{(1 - \alpha_{\rm A})(1 + \alpha)\alpha_{\rm A}^+\alpha^+}
{(1 - \alpha_{\rm A}^+)( 1 + \alpha^+)\alpha_{\rm A} \alpha }
\left(\frac{T}{T_+}\right)^{\! \frac{5}{2}} e^{-\frac{T_{\rm A}}{T} + \frac{T_{\rm A}}{T_+}}\nonumber\\
&+ \frac{(1 - \alpha_{\rm A}^+)\alpha_{\rm A}\alpha}
{(1 - \alpha_{\rm A})\alpha_{\rm A}^+ \alpha^+}
\left(\frac{T}{T_+}\right)^{\! -\frac{3}{2}} e^{-\frac{T_{\rm A}}{T_+} + \frac{T_{\rm A}}{T}} - \frac{T(1 + \alpha)}{T_+(1 + \alpha^+)} - 1 . \label{eq:GGG}
\end{align}
By denoting
$
  \Lambda^+ =  \frac{(1 - \alpha_{\rm A}^+)\alpha_{\rm A}\alpha}
{(1 - \alpha_{\rm A})\alpha_{\rm A}^+ \alpha^+},
$
we have $\frac{v}{v_+} = \Lambda^+\left(\frac{T}{T_+}\right)^{\! -\frac{3}{2}} e^{-\frac{T_{\rm A}}{T_+} + \frac{T_{\rm A}}{T}}$ and
\begin{multline*}
   G(\alpha_{\rm A}, T) 
  =  K_+\left(\frac{T}{T_+}\right)^{\! \frac{5}{2}} e^{-\frac{T_{\rm A}}{T} + \frac{T_{\rm A}}{T_+}}
+ \Lambda^+\left(\frac{T}{T_+}\right)^{\! -\frac{3}{2}} e^{-\frac{T_{\rm A}}{T_+} + \frac{T_{\rm A}}{T}} \\
  - \frac{T(1 + \alpha)}{T_+(1 + \alpha^+)} - \frac{M(u_A - u_+)^2}{RT_+(1 + \alpha^+)} - 1 .
\end{multline*}
\par
Since
\begin{equation}\label{eq:dPsi/dalpha}
  \left(\frac{\partial  \Lambda^+}{\partial T}\right)_{\! \alpha_{\rm A}} 
= \frac{\Lambda^+Q_{\rm BA}}{\alpha T}, \qquad 
   \left(\frac{\partial \Lambda^+}{\partial \alpha_{\rm A}}\right)_{\! T} 
= \frac{\Lambda^+}{q_{\rm A}}\left(1 +  \frac{q}{\alpha}\right),
\end{equation}
 we find together with \eqref{eq:dOmega/dalpha} that
\begin{align}
 & q_{\rm A}\left(\frac{\partial G}{\partial \alpha_{\rm A}}\right)_{\! T}
 =  - \left[1 + \frac{q }{\alpha(1 + \alpha)}\right]\left(\frac{p}{p_+}\right)
 + \left(1 + \frac{q }{\alpha}\right)\left(\frac{v}{v_+}\right)
  - \frac{Tq }{T_+(1 + \alpha^+)}, \label{eq:dG/d alpha}\\
    & T\frac{\partial G}{\partial T} (\alpha_{\rm A}, T)
=  \frac{p}{p_+}\left[\frac{5}{2} + \frac{T_{\rm A}}{T} -  \frac{Q_{\rm BA}}{\alpha(1 + \alpha)}\right]
- \frac{v}{v_+}\left(\frac{3}{2} + \frac{T_{\rm A}}{T} -  \frac{Q_{\rm BA}}{\alpha }\right)\nonumber\\
& \hspace{12.5ex} - \frac{T(1 + \alpha)}{T_+(1 + \alpha^+)} -  \frac{TQ_{\rm BA}}{T_+(1 + \alpha^+)}\label{eq:TdG/dT_1}\\
&=  \left(\frac{p}{p_+}- \frac{v}{v_+}\right)\left(\frac{3}{2} + \frac{T_{\rm A}}{T} -  \frac{Q_{\rm BA}}{\alpha}\right) 
 +\left[ \frac{p}{p_+} - \frac{T(1 + \alpha)}{T_+(1 + \alpha^+)} \right] \left(1 + \frac{Q_{\rm BA}}{1 + \alpha}\right),\label{eq:TdG/dT_2}
\end{align}
where the identity $\frac{1}{\alpha} - \frac{1}{\alpha(1 + \alpha)} = \frac{1}{1 + \alpha}$ is used.
\par 
\begin{proposition}\label{prop:set G=0}
For every  $\alpha_{\rm A},$ there are at least two values $T^{(\pm)} = T^{(\pm)}(\alpha_{\rm A})$, with $T^{(-)} < T^{(+)}$, such that $G\left(\alpha_{\rm A}, T^{(\pm)}\right)= 0.$ Moreover, $p^{(+)} > p_+ > p^{(-)}.$
\end{proposition}
\begin{proof}
  By virtue of the assumption: $T_{\rm A} < T_{\rm B} \leq  2T_{\rm A}$, we find that
\begin{equation}\label{eq:TA/T-QAB/alpha}
    \frac{T_{\rm A}}{T} -  \frac{Q_{\rm BA}}{\alpha} =  \frac{T_{\rm A}}{T} -  \frac{(1 - \beta)(T_{\rm B} - T_{\rm A})q_{\rm B}}{\alpha T} \geq  \frac{2T_{\rm A}}{T} - \frac{T_{\rm B}}{T} \geq 0. 
\end{equation}
Hence
$$
\left[\frac{\partial }{\partial T}\left(\frac{p}{p_+}\right)\right]_{\! \alpha_{\rm A}}
=
\frac{p}{p_+ T}\left[\frac{5}{2} + \frac{T_{\rm A}}{T} -  \frac{(1 - \beta)(T_{\rm B} - T_{\rm A})q_{\rm B}}{\alpha(1 + \alpha) T}\right]
\geq  \frac{5p}{2p_+T}
$$
which shows that the function $T\mapsto \frac{p}{p_+}$ is strictly increasing and valued in $[0,\infty)$. Then, for every $\alpha_{\rm A} > 0$ there exists a unique $T_{\ast} = T_{\ast}(\alpha_{\rm A})$ $(T_\ast(\alpha^+)=T_+)$ such that, by \eqref{eq:p_-/p_+},
$$
\frac{p_{\ast}}{p_+} = \frac{(1 - \alpha_{\rm A})(1 + \alpha)\alpha_{\rm A}^+\alpha^+}
{(1 - \alpha_{\rm A}^+)( 1 + \alpha^+)\alpha_{\rm A} \alpha }
\left(\frac{T_{\ast}}{T_+}\right)^{\! \frac{5}{2}} e^{-\frac{T_{\rm A}}{T_{\ast}} + \frac{T_{\rm A}}{T_+}} = 1,
$$
for $p_*=p\left(\alpha_{\rm A},T_*(\alpha_{\rm A})\right).$ Thus 
$G\left(\alpha_{\rm A},T_{\ast}(\alpha_{\rm A})\right) = - \frac{ (u - u_0)^2}{a^2T_0(1 + \alpha_0)}  < 0.$
\par
For every fixed $\alpha_{\rm A}$ we have that $G(\alpha_{\rm A}, T) \to \infty$ for both $T \to 0$ and $T\to\infty$.
We conclude that there are at least two values $T^{(\pm)} = T^{(\pm)}(\alpha_{\rm A})$, with $T_- < T_*<T_+$, such that $G\left(\alpha_{\rm A}, T^{(\pm)}\right)= 0.$  Note that $\frac{p^{(+)}}{p_+} > \frac{p_{\ast}}{p_+} = 1$ and $\frac{p^{(-)}}{p_+} < \frac{p_{\ast}}{p_+} = 1.$
\end{proof}
We call the set of $(\alpha_{\rm A}, T^{(+)})$ the {\it compressive part\/} of $G\left(\alpha_{\rm A}, T\right)= 0$ and the set  of  $(\alpha_{\rm A}, T^{(-)})$ its  {\it expansive part\/} .
\paragraph{Asymptotics and Monotonicity:}
First we shall show the asymptotic behavior of the set $G(\alpha_{\rm A},T^{(\pm)})=0$ for $\alpha_{\rm A}$ close to $1$. 
\begin{proposition}\label{prop:G=0 asymptotics}
  Suppose that $0 < \alpha_{\rm A}^+ < 1.$ Then we have
\begin{equation}\label{eq:asymptotics T_+(alpha)}
  1 - \alpha_{\rm A} \sim \frac{1 - \alpha_{\rm A}^+}{\alpha_{\rm A}^+\alpha^+}e^{- \frac{T_{\rm A}}{T_+}}\left(\frac{T}{T_+}\right)^{\! -\frac{3}{2}}
  \quad for \quad T \to \infty
\end{equation}
for the compressive part and
\begin{equation}\label{eq:asymptotics T_-(alpha)}
  1 - \alpha_{\rm A} \sim \frac{(1 - \alpha_{\rm A}^+)(1 + \alpha^+)}{2\alpha_{\rm A}^+\alpha^+}e^{- \frac{T_{\rm A}}{T_+}} \left(\frac{T}{T_+}\right)^{\! -\frac{5}{2}}
    \quad for \quad T \to \infty
\end{equation}
for the expansive part. Hence, for $\alpha_{\rm A}$ close to $1,$ both parts of $G(\alpha_{\rm A}, T) = 0$ constitute strictly increasing  curves in $(T, \alpha_{\rm A})$ plane.
\end{proposition}
\begin{proof}
It is easy to see from \eqref{eq:GGG} and compatibility condition that: if $\alpha_{\rm A}\to 1$ on  $G(\alpha_{\rm A},T)=0,$ then $\alpha_{\rm B}\to 1$ and  $T \to \infty.$  We set
$
   1 - \alpha_{\rm A} \sim K\left(\frac{T}{T_+}\right)^{\! -\mu}
$  
with  $\mu > 0$. By \eqref{eq:GGG}
\begin{align*}
&\frac{2\alpha_{\rm A}^+ \alpha^+K}{(1 - \alpha_{\rm A}^+)(1 + \alpha^+)}\left(\frac{T}{T_+}\right)^{\!\frac{5}{2} - \mu}e^{ \frac{T_{\rm A}}{T_+}}   +  \frac{1 - \alpha_{\rm A}^+}{\alpha_{\rm A}^+\alpha^+K}\left(\frac{T_0}{T}\right)^{\!\frac{3}{2} - \mu}e^{-\frac{T_{\rm A}}{T_+}}\\
 & \sim    \frac{2T}{T_+(1 + \alpha^+)} + 1 + \frac{M (u_A - u_+)^2}{RT_+(1 + \alpha^+)} .
\end{align*} 
If $T^{\frac{5}{2} - \mu}$ and $T^{\mu - \frac{3}{2}}$ are equally large as $T \to \infty,$ we have $\mu = 2.$ Then both terms are $O(1)T^{\frac{1}{2}},$ which is a contradiction. In the case $\frac{5}{2} - \mu = 1,$ we have $\mu = \frac{3}{2}$ and $K = \frac{1 - \alpha_{\rm A}^+}{\alpha_{\rm A}^+\alpha^+}e^{- \frac{T_{\rm A}}{T_+}}, $ which is the compressive part, and obtain \eqref{eq:asymptotics T_+(alpha)}.
\par
Otherwise, if $\mu - \frac{3}{2} = 1,$ we have $\mu = \frac{5}{2}$ and  $K = \frac{(1 - \alpha_{\rm A}^+)(1 + \alpha^+)}{2\alpha_{\rm A}^+\alpha^+}e^{- \frac{T_{\rm A}}{T_+}},$ which is the expansive part. In both cases monotonicity with respect to $\alpha_{\rm A}$ close to $1$ is obvious.
\end{proof}
\paragraph{Smoothness of the Comressive Part:}
\begin{theorem}\label{thm:G=0}
For any fixed $(\alpha_{\rm A}^+, T_+),$   the compressive part of $G(\alpha_{\rm A},T) = 0,$ constitutes a single differentiable curve.
\end{theorem}
\begin{proof}
  Since
  $
  G(\alpha_{\rm A}, T) = 0
  $
  is written as
  $\ 
   \frac{p}{p_+} - \frac{T(1 + \alpha)}{T_+(1 + \alpha^+)}
=  1 -  \frac{v}{v_+} + \frac{M(u_A - u_+)^2}{RT_+(1 + \alpha^+)}
$
    then it follows from \eqref{eq:TdG/dT_2} and \eqref{eq:TA/T-QAB/alpha} that 
\begin{align}
  T\frac{\partial G}{\partial T} (\alpha_{\rm A}, T)
  &=  \left(\frac{p}{p_+}- \frac{v}{v_+}\right)\left(\frac{3}{2} + \frac{T_{\rm A}}{T} -  \frac{Q_{\rm BA}}{\alpha}\right) \nonumber\\
&+  \left(1 + \frac{Q_{\rm BA}}{1 + \alpha}\right)\left[\left(1 -  \frac{v}{v_+}\right) + \frac{M(u_A - u_+)^2}{RT_+(1 + \alpha^+)}\right] > 0 \label{eq:dG/dT}
\end{align} 
on the compressive part  $p > p_+, \, v < v_+.$ 
Thus we have proved that: in a neighbourhood of every state $(\alpha_{\rm A}, T^{(+)}),$ the compressive part of $G(\alpha_{\rm A},T) = 0$ is a graph of a differentiable function of $\alpha_{\rm A}.$ Consequently, the compressive part constitutes a differentiable curve.
\end{proof}
\paragraph{Solution to the Shock Tube Problem:}
\begin{theorem}[Intersection with Hugoniot loci]\label{thm:intersection}
Fix $(T_+,\alpha_{\rm A}^+)$ and $u_+\ne u_A$. 
Then, in the region $\alpha_{\rm A}^+ < \alpha_{\rm A}  < 1,$ there is at least one  intersection point of the compressive part of $G(\alpha_{\rm A},  T) = 0$ and the thermodynamic Hugoniot locus \eqref{eq:Rankine-Hugoniot} of $(T_+, \alpha_{\rm A}^+)$. 
\end{theorem}
\begin{proof}
  Let $\alpha_{\rm A} = \alpha_{\rm A}(T)$ denote thermodynamic Hugoniot locus of $(T_+, \alpha_{\rm A}^+)$ in $(\alpha_{\rm A},T)$-plane. Clearly, $G(\alpha_{\rm A}^+, T_+) = - \frac{M(u_A - u_+)^2}{RT_+(1 + \alpha^+)} < 0$ by \eqref{eq:GGG}, while  $\alpha_{\rm A}^+  = \alpha_{\rm A}(T_+).$ Through the proof of Proposition \ref{prop:set G=0}, we have found that the graph of $T = T^{(+)}(\alpha_{\rm A})$ is located above the thermodynamic part the Hugoniot locus of $(T_+, \alpha_{\rm A}^+)$ near $\alpha_{\rm A} = \alpha_{\rm A}^+.$
\par
On the other hand, by \eqref{eq:asymptotics T(alpha)} the Hugoniot locus of $(\alpha_{\rm B}, T_B)$ takes an aymptotic form 
$
    	1 - \alpha_{\rm A} \sim \frac{4(1 - \alpha^+)}{\alpha_{\rm A}^+\alpha^+}\left(\frac{T}{T_+}\right)^{\!-\frac{3}{2}}e^{-\frac{T_{\rm A}}{T_+}}
$
and by \eqref{eq:asymptotics T_+(alpha)}
$T_{+}(\alpha)$ has 
$
 1 - \alpha_{\rm A} \sim \frac{1 - \alpha_{\rm A}^+}{\alpha_{\rm A}^+\alpha^+}e^{- \frac{T_{\rm A}}{T_+}}\left(\frac{T}{T_+}\right)^{\! -\frac{3}{2}}.
$ as $\alpha_{\rm A} \to 1.$
Thus we conclude that the graph of  $T = T^{(+)}(\alpha_{\rm A})$ is located under the Hugoniot locus as $\alpha_{\rm A} \to 1,$ which proves the assertion. 
\end{proof}
\paragraph{Uniqueness of the Solution:}
The aim of this subsection is to prove that the solution obtained in Theorem \ref{thm:intersection} is unique.
\begin{theorem}[Uniqueness of the intersection point]\label{thm:1ness intersection}
Fix $(\alpha_{\rm A}^+, T_+)$ and $u_+\ne u_A$. 
Then, in the region $\alpha_{\rm A}^+ < \alpha_{\rm A} < 1$ the intersection point of the compressive part of $G(\alpha_{\rm A},  T) = 0$ and the thermodynamic Hugoniot locus \eqref{eq:Rankine-Hugoniot} of $(T_+, \alpha_{\rm A}^+)$ is unique. 
\end{theorem}
\begin{proof}
Since the thermodynamic Hugoniot locus is the graph of $\alpha_{\rm A} = \alpha_{\rm A}(T),$ the solution is a zero of $G(\alpha_{\rm A}(T),T)$ (denoted by $T = T_{\ast}$) and proof will be completed by showing
$$
  \frac{dG}{dT} (\alpha_{\rm A}(T_{\ast}),T_{\ast})  
  = \frac{\left[\left(\frac{\partial G}{\partial T}\right)_{\rm \!A} {\left(\frac{\partial H}{\partial \alpha_{\rm A}}\right)_{\! T}}
    - \left(\frac{\partial G}{\partial \alpha_{\rm A}}\right)_{\! T}\left(\frac{\partial H}{\partial T}\right)_{\rm \!A}\right]_{(\alpha_{\rm A}(T_{\ast}),T_{\ast})} }{\left(\frac{\partial H}{\partial \alpha_{\rm A}}\right)_{\! T}(\alpha_{\rm A}(T_{\ast}),T_{\ast}) } >  0.
$$
Recall that $\left(\frac{\partial H}{\partial \alpha_{\rm A}}\right)_{\! T} > 0.$ By substituting the expressions \eqref{eq:dH/dT 2}, \eqref{eq:dH/dalpha 2},  \eqref{eq:dG/d alpha},  \eqref{eq:TdG/dT_1} into the above, the numerator
is written as $ \frac{T_+\left(1 + \alpha^+\right)}{q_{\rm A}T(1 + \alpha)}$ times of 
\begin{multline*}
  \left|
    \begin{array}{c}
      \left[\frac{5}{2} + \frac{T_{\rm A}}{T} -  \frac{Q_{\rm BA}}{\alpha(1 + \alpha) }\right]\frac{p}{p_+}
- \left(\frac{3}{2} + \frac{T_{\rm A}}{T} -  \frac{Q_{\rm BA}}{\alpha }\right)\frac{v}{v_+}
 - \frac{T(1 + \alpha)}{T_+(1 + \alpha^+)} -  \frac{TQ_{\rm BA}}{T_+(1 + \alpha^+)} \\
     - (1 + \alpha)\left[1 + \frac{q}{\alpha(1 + \alpha)}\right]\frac{p}{p_+}
+ (1 + \alpha)\left(1 + \frac{q}{\alpha}\right)\frac{v}{v_+}
  - \frac{qT(1 + \alpha) }{T_+(1 + \alpha^+)}
    \end{array}
\right.\\
\left.
  \begin{array}{c}
   \! -\! \left[\frac{5}{2}\! + \!\frac{T_{\rm A}}{T}
 \!-\! \frac{ Q_{\rm BA}}{\alpha(1 + \alpha)}\right]\frac{p}{p_+}\!-\! \left(\frac{3}{2}\!+\! \frac{T_{\rm A}}{T}
\!-\! \frac{ Q_{\rm BA}}{\alpha}\right)\frac{v}{v_+}
\!+\! \frac{4T(1 + \alpha)}{T_+\left(1 + \alpha^+\right)}\!\left[1\! +\! \frac{ Q_{\rm BA}}{1 + \alpha}\!\left(1 \!+\! \frac{T_{\rm B}}{2T}\right)\right]\\
 (1 + \alpha)\left[1 + \frac{q}{\alpha(1 + \alpha)}\right]\frac{p}{p_+}
+ (1 + \alpha)\left(1 + \frac{q}{\alpha}\right)\frac{v}{v_+}
+ \frac{4T(1 + \alpha)}{T_+\left(1 + \alpha^+\right)}\left(q  + \frac{Q_T}{2}\right)
  \end{array}
\right|.
\end{multline*}
As shown in section \ref{sec:computation},  its final form will be
\begin{align} 
& \textstyle  3 (1 + \alpha)\left[1 + \frac{q}{\alpha(1 + \alpha)}\right]
\left(\frac{p}{p_+} - 1\right)
+ 5 (1 + \alpha)\left(1 + \frac{q}{\alpha}\right)\left(1 - \frac{v}{v_+}\right)
\nonumber \\ 
&  \textstyle   + 8\left\{
 q \left[\left(\frac{3}{2}  +  \frac{T_{\rm A}}{T}\right)\left(\frac{5}{2} + \frac{T_{\rm A}}{T}\right) - \frac{T_{\rm A}}{4T}\right]
 +  \left[\frac{q\left(T_{\rm B} - T_{\rm A}\right)}{T} - Q_{\rm BA}\right]\frac{Q_{\rm BA}}{\alpha(1 + \alpha) }\right.\nonumber\\
&  \textstyle \left. + \left(\frac{15}{4} + \frac{T_{\rm A} + T_{\rm B}}{T}\right)Q_{\rm BA} \rule{0ex}{2.5ex} 
 \right\} \! \left[\frac{T(1 + \alpha)}{T_+(1 + \alpha^+)}  - 1\right]
 \nonumber\\
& \textstyle +  6 (1 + \alpha)\left[1 + \frac{q}{\alpha(1 + \alpha)}\right]
\left(\frac{p}{p_+} - 1\right)
+ 5 (1 + \alpha)\left(1 + \frac{q}{\alpha}\right)\left(1 - \frac{v}{v_+}\right)
\nonumber \\
&  \textstyle    + 2\left\{
 q \left(\frac{3}{2}  +  \frac{T_{\rm A}}{T}\right)\left(\frac{5}{2} + \frac{T_{\rm A}}{T}\right)
 +  \left[\frac{q\left(T_{\rm B} - T_{\rm A}\right)}{T} - Q_{\rm BA}\right]\frac{Q_{\rm BA}}{\alpha(1 + \alpha) }\right.\nonumber\\ 
&  \textstyle  \left.  + \left(4 + \frac{T_{\rm A} + T_{\rm B}}{T}\right)Q_{\rm BA} \right\} \! \frac{\left[\beta T_{\rm A}(\alpha_{\rm A}  -\alpha_{\rm A}^+) + (1 - \beta)T_{\rm B}(\alpha_{\rm B} -\alpha_{\rm B}^+)\right] }{T_+(1 + \alpha^+)}\nonumber\\
& \textstyle + 2 \left(\frac{v}{v_+}\right)  \left[\frac{q\left(T_{\rm B} - T_{\rm A}\right)}{T} -  Q_{\rm BA}\right]  \left(\frac{p}{p_+} - 1\right)
\frac{Q_{\rm BA}}{1 + \alpha}.\label{eq:dG/dT final}
\end{align}
Notice that
\begin{equation}\label{eq:q(T_B - T_A)/T - Q_AB}
    \frac{q\left(T_{\rm B} - T_{\rm A}\right)}{T} -  Q_{\rm BA}
    = \frac{\beta\left(T_{\rm B} - T_{\rm A}\right)q_{\rm A}}{T} > 0,
\end{equation}
which shows that the expession \eqref{eq:dG/dT final} is positive and the theorem follows.
\end{proof}

\section{Conclusions and Discussions}\label{sec:conclusions}
In this paper, we have studied a model system for  macroscopic motion of an ionized  gas which is a mixture of two monatomic gas A and B; the mixture ratio is $\beta :1 - \beta.$ This model system is proposed by
\cite{Fukuda-Okasaka-Fujimoto} and consists of three conservation laws  in one space dimension together with the first and second law of thermodynamics which are supplemented by an equation of state and two more thermodynamic equations  called Saha's laws. We have assumed that the interaction potential energies and effects of collisions between charged particles are negligible and the local thermodynamic equilibrium is attained. We further assume that $T_{\rm B}:$  the first ionization temperature of the gas B is higher than $T_{\rm A}:$  that of the gas A and $2T_{\rm A} \geq T_{\rm B}.$ Note that A: hydrogen and B: helium satisfy these assumptions. 
\par
The physical entropy functions are constructed and it is remarkable that they are expressed in terms of elementary functions.  Saha's two equations bring about a compatibility condition involving $\alpha_{\rm A}, \alpha_{\rm B}$ and $T.$ It is shown that  $\alpha_{\rm B}$ is a differentiable function of $\alpha_{\rm A}$ and $T$ whose graph constitutes  the thermodynamic state space. We propose that $(T, \alpha_{\rm A})$ is a suitable pair of independent thermodynamic state variables.
\par
The system of conservation laws is shown to constitute a strictly hyperbolic system, which  implies that the initial-value problem is well-posed locally in time for sufficiently smooth initial data. Characteristic fields are computed and geometric properties are studied: unlike the polytropic (non-ionized) case, the convexity (genuine nonlinearity) of the  forward and backward  characteristic fields of the system is lost and the set where this happens is determined in a neighbourhood of $T = \alpha_{\rm A}  = 0.$ Whole set is located in a finite region in  $(T, \alpha_{\rm A})$ plane but it is difficult  to get its full picture by purely mathematical reasoning; only  pictures by numerical computation are presented.
\par
A detailed study of the thermodynamic Hugoniot locus  is performed. For every $T > 0,$ there is a unique $0 < \alpha_{\rm A} < 1$ satisfying the thermodynamic Rankine-Hugoniot condition and  $\alpha_{\rm A}$ is a smooth function of $T.$ Hence the Hugoniot locus is a smooth graph in the $(T, \alpha_{\rm A})$ plane. While the thermodynamic Hugoniot locus is monotone in $(T, \alpha)$ plane in a single monatomic case,  for the mixed monatomic case it is shown that: if $\beta$ is sufficiently small, then it loses monotonicity at some base state. Thus the degree of ionization does not always increase across the shock front, even if the temperature increases. However the pressure is actually proved to increase as the temperature increases which ensures that $T >T_+$ is the admissible branch.
\par
In order to fit the mathematical data to ordinary circumstances, an approximation of thermodynamic Hugoniot loci is proposed and proved that, for small $T,$ it is limited in a  \lq\lq classical\rq\rq\ region where the forward and backward characteristic fields are convex (genuinely nonlinear) and the physical entropy increases along the admissible branch.  We expect that actual experiments  are usually performed in such a classical region.
\par
These results are applied to the mathematical analysis of the shock tube problem:  existence and uniqueness of the solution are established, which provides  a rigorous mathematical basis to the physical phenomena observed and reported in \cite{Fukuda-Okasaka-Fujimoto}, \cite{Fukuda_2}, \cite{Fukuda_3}.

\appendix
\section{Computation of \eqref{eq:dG/dT final}}\label{sec:computation}
Substitution of \eqref{eq:dH/dT 2}, \eqref{eq:dH/dalpha 2},  \eqref{eq:dG/d alpha}, \eqref{eq:TdG/dT_1} into  $q_{\rm A}\left[\left(\frac{\partial G}{\partial T}\right)_{\! \alpha_{\rm A}} \left(\frac{\partial H}{\partial \alpha_{\rm A}}\right)_{\! T}
-   \left(\frac{\partial G}{\partial \alpha_{\rm A}}\right)_{\! T}\left(\frac{\partial H}{\partial T}\right)_{\! \alpha_{\rm A}}\right]$ gives
\begin{multline*}
  \textstyle
  \left|
    \begin{array}{cc}
- \left(3 + \frac{2T_{\rm A}}{T} -  \frac{2Q_{\rm BA}}{\alpha }\right) 
\!&\!    - \left[5 + \frac{2T_{\rm A}}{T}  - \frac{2Q_{\rm BA}}{\alpha(1 + \alpha)}\right]
\\
 (1 + \alpha)\left(1 + \frac{q}{\alpha}\right) 
\!&\!
 (1 + \alpha)\left[1 + \frac{q}{\alpha(1 + \alpha)}\right]
  \end{array}
\right|\\
 +  \textstyle  \frac{p}{p_+}\left|
    \begin{array}{cc}
      \frac{5}{2} + \frac{T_{\rm A}}{T} -  \frac{Q_{\rm BA}}{\alpha(1 + \alpha) }
\!\!&\!\! 3+  \frac{3 Q_{\rm BA}}{1 + \alpha} +  \frac{2T_{\rm B}Q_{\rm BA}}{T(1 + \alpha)}
 \\
     - (1 + \alpha)\left[1 + \frac{q}{\alpha(1 + \alpha)}\right]
\!\!&\!\! (3q  + 2Q_T)
  \end{array}
\right|\\
  \textstyle - \frac{v}{v_+} \left|
    \begin{array}{cc}
\frac{3}{2} + \frac{T_{\rm A}}{T} -  \frac{Q_{\rm BA}}{\alpha } 
& 5 + \frac{5Q_{\rm BA}}{1 + \alpha}  + \frac{2T_{\rm B} Q_{\rm BA}}{T(1 + \alpha)}\\
 -(1 + \alpha)\left(1 + \frac{q}{\alpha}\right)
&  (5q  + 2Q_T)
  \end{array}
    \right| \\
-\frac{T(1 + \alpha)}{T_+(1 + \alpha^+)}\left|
    \begin{array}{cc}
 1 +  \frac{Q_{\rm BA}}{1 + \alpha} 
&  4\left[1 + \frac{ Q_{\rm BA}}{1 + \alpha}\left(1 + \frac{T_{\rm B}}{2T}\right)\right]\\
  q
& 4\left(q  + \frac{Q_T}{2}\right)
  \end{array}
    \right|\\
  \textstyle =  - (1 + \alpha)\left[1 + \frac{q}{\alpha(1 + \alpha)}\right]
\left(3 + \frac{2T_{\rm A}}{T} -  \frac{2Q_{\rm BA}}{\alpha }\right)
+  (1 + \alpha)\left(1 + \frac{q}{\alpha}\right)\left[5 + \frac{2T_{\rm A}}{T}  - \frac{2Q_{\rm BA}}{\alpha(1 + \alpha)}\right]\\
  \textstyle +2 \left(\frac{p}{p_+}\right)
      \left(\frac{3}{2}q  + Q_T\right)\left[ \frac{5}{2} + \frac{T_{\rm A}}{T} -  \frac{Q_{\rm BA}}{\alpha(1 + \alpha) }\right]\\
  \textstyle + (1 + \alpha)\left(\frac{p}{p_+}\right)\left[1 + \frac{q}{\alpha(1 + \alpha)}\right]
\left[3+  2\left(\frac{3}{2} +  \frac{T_{\rm B}}{T}\right)\frac{Q_{\rm BA}}{1 + \alpha}\right]\\
- 2\left(\frac{v}{v_+}\right)
\left(\frac{5}{2}q  + Q_T\right)\left[\frac{3}{2} + \frac{T_{\rm A}}{T} -  \frac{Q_{\rm BA}}{\alpha } \right]\\
  \textstyle - (1 + \alpha)\left(\frac{v}{v_+}\right)\left(1 + \frac{q}{\alpha}\right)
\left[ 5 + 2\left(\frac{5}{2}  + \frac{T_{\rm B} }{T}\right)\frac{Q_{\rm BA}}{1 + \alpha}\right]\\
-\frac{T(1 + \alpha)}{T_+(1 + \alpha^+)}
\left[
2Q_T\left( 1 +  \frac{Q_{\rm BA}}{1 + \alpha} \right)- \frac{2qT_{\rm B}Q_{\rm BA}}{T(1 + \alpha)}
\right].
\end{multline*}
\par
Notice that
$$
  qQ_T -  \frac{qT_{\rm A}}{T} = \frac{\beta q_{\rm A}T_{\rm A} +  (1 - \beta)q_{\rm B}T_{\rm B}}{T} - \frac{T_{\rm A}}{T}
  = \frac{(1 - \beta)q_{\rm B}(T_{\rm B} - T_{\rm A})}{T} = Q_{\rm BA}.
$$
Hence we have
  $ qQ_T = \frac{qT_{\rm A}}{T} + Q_{\rm BA}.$
By this formula, the above expression is found to be
\begin{multline*}
 \textstyle  3 (1 + \alpha)\!\left[1 + \frac{q}{\alpha(1 + \alpha)}\right]\!
\left(\frac{p}{p_+} - 1\right)
+ 5 (1 + \alpha)\left(1 + \frac{q}{\alpha}\right)\left(1 - \frac{v}{v_+}\right)
 \\
 \textstyle     + 2\left[
 q \left(\frac{3}{2}  +  \frac{T_{\rm A}}{T}\right)\!\left(\frac{5}{2} + \frac{T_{\rm A}}{T}\right)
 +  q\left(\!\frac{T_{\rm B} - T_{\rm A}}{T}\!\right)\frac{Q_{\rm BA}}{\alpha }
    + \left(4 + \frac{T_{\rm A} + T_{\rm B}}{T} -  \frac{Q_{\rm BA}}{\alpha }\right)\! Q_{\rm BA} \right]\!\left(\frac{p}{p_+} - \frac{v}{v_+} \right)\\
  \textstyle  - 2 q\left(\frac{p}{p_+}\right)\!\left(\!\frac{T_{\rm B} - T_{\rm A}}{T}\!\right)\frac{Q_{\rm BA}}{1 + \alpha } + 2\left(\frac{p}{p_+}\right) \frac{Q_{\rm BA}^2}{1 + \alpha}\\
   \textstyle  -2q \left(\frac{T_{\rm A}}{T} + \frac{Q_{\rm BA}}{q}\right)
\left[\frac{T(1 + \alpha)}{T_+(1 + \alpha^+)}
\left( 1 +  \frac{Q_{\rm BA}}{1 + \alpha} \right) -1\right] + \left[\frac{T(1 + \alpha)}{T_+(1 + \alpha^+)}\right]\frac{2qT_{\rm B}Q_{\rm BA}}{T(1 + \alpha)}
\end{multline*}
\vspace{-5ex}
\begin{multline*}
   \textstyle = 3 (1 + \alpha)\!\left[1 + \frac{q}{\alpha(1 + \alpha)}\right]\!
\left(\frac{p}{p_+} - 1\right)
+ 5 (1 + \alpha)\left(1 + \frac{q}{\alpha}\right)\!\left(1 - \frac{v}{v_+}\right)
 \\
  \textstyle     + 2\left[
 q \left(\frac{3}{2} \! +\!  \frac{T_{\rm A}}{T}\right)\!\left(\frac{5}{2}\! +\! \frac{T_{\rm A}}{T}\right)
 \!+\!  q\left(\!\frac{T_{\rm B}\! - \!T_{\rm A}}{T}\!\right)\frac{Q_{\rm BA}}{\alpha(1 \!+\! \alpha) }
   \! +\! \left(4\! +\! \frac{T_{\rm A}\! +\! T_{\rm B}}{T}\! - \! \frac{Q_{\rm BA}}{\alpha(1\! +\! \alpha) }\right)\!Q_{\rm BA} \right]\!\left(\frac{p}{p_+}\! -\! \frac{v}{v_+} \right)\\
   \textstyle  - 2 q\left(\frac{v}{v_+}\right)\left(\frac{T_{\rm B} - T_{\rm A}}{T}\right)\frac{Q_{\rm BA}}{1 + \alpha } + 2\left(\frac{v}{v_+}\right) \frac{Q_{\rm BA}^2}{1 + \alpha}\\
   \textstyle  -2q \left(\frac{T_{\rm A}}{T} + \frac{Q_{\rm BA}}{q}\right)
\left[\frac{T(1 + \alpha)}{T_+(1 + \alpha^+)} - 1\right]
- \frac{2T(1 + \alpha)}{T_+(1 + \alpha^+)}\frac{Q_{\rm BA}^2}{1 + \alpha} + \left[\frac{T(1 + \alpha)}{T_+(1 + \alpha^+)}\right]\frac{2q\left(T_{\rm B} - T_{\rm A}\right)Q_{\rm BA}}{T(1 + \alpha)}.
\end{multline*}
\par
Note that
$
   \frac{pv}{p_+v_+} = \frac{T(1 + \alpha)}{T_+(1 + \alpha^+)}.
$
Then
\begin{align*}
  \textstyle  2q\left[\frac{T(1 + \alpha)}{T_+(1 + \alpha^+)}\right]\frac{T_{\rm B} - T_{\rm A}}{T}\frac{Q_{\rm BA}}{1 + \alpha}
- 2 q\left(\frac{v}{v_+}\right)\frac{T_{\rm B} - T_{\rm A}}{T}\frac{Q_{\rm BA}}{1 + \alpha}
&  \textstyle = 2 q\left(\frac{v}{v_+}\right)\frac{\left(T_{\rm B} - T_{\rm A}\right)Q_{\rm BA}}{T(1 + \alpha) }\left(\frac{p}{p_+} - 1\right) \geq 0,\\
  \textstyle -  \frac{2T(1 + \alpha)}{T_+(1 + \alpha^+)}\frac{Q_{\rm BA}^2}{1 + \alpha}
+ 2\left(\frac{v}{v_+}\right) \frac{Q_{\rm BA}^2}{1 + \alpha}
& \textstyle = - 2\left(\frac{v}{v_+}\right)\left(\frac{p}{p_+} - 1 \right) \frac{Q_{\rm BA}^2}{1 + \alpha},
\end{align*}
and we conclude that the above expression is equal to
\begin{multline*}
 \textstyle  3 (1 + \alpha)\!\left[1 + \frac{q}{\alpha(1 + \alpha)}\right]\!
\left(\frac{p}{p_+} - 1\right)
+ 5 (1 + \alpha)\left(1 + \frac{q}{\alpha}\right)\left(1 - \frac{v}{v_+}\right)
 \\
  \textstyle    + 2\left[
 q \left(\frac{3}{2}  \!+\!  \frac{T_{\rm A}}{T}\right)\!\left(\frac{5}{2}\! +\! \frac{T_{\rm A}}{T}\right)
\! +\!  q\left(\!\frac{T_{\rm B} \!-\! T_{\rm A}}{T}\!\right)\!\frac{Q_{\rm BA}}{\alpha(1 + \alpha) }  
 \! +\! \left(4\! +\! \frac{T_{\rm A} + T_{\rm B}}{T}\! - \! \frac{Q_{\rm BA}}{\alpha(1 + \alpha) }\right)\!Q_{\rm BA} \right]\!\left(\frac{p}{p_+} \!-\! \frac{v}{v_+} \right)
  \\
   \textstyle -2q \left(\frac{T_{\rm A}}{T} + \frac{Q_{\rm BA}}{q}\right)
  \left[\frac{T(1 + \alpha)}{T_+(1 + \alpha^+)}  - 1\right] 
+ 2 \left(\frac{v}{v_+}\right)\left[\frac{q\left(T_{\rm B} - T_{\rm A}\right)}{T} -  Q_{\rm BA}\right]\left(\frac{p}{p_+} - 1\right)
\frac{Q_{\rm BA}}{1 + \alpha}.
\end{multline*}
\par
Recall that the thermodynamic Rankine-Hugoniot condition is  equivalent to 
\begin{align*}
  \frac{p}{p_+} - \frac{v}{v_+}
& = \frac{4T(1 + \alpha)}{T_+(1 + \alpha^+)} -4  + \frac{2\left[\beta T_{\rm A}(\alpha_{\rm A}  -\alpha_{\rm A}^+) + (1 - \beta)T_{\rm B}(\alpha_{\rm B} -\alpha_{\rm B}^+)\right] }{T_+(1 + \alpha^+)}.
\end{align*}
Thus $\textstyle q_{\rm A}\left[\left(\frac{\partial G}{\partial T}\right)_{\! \alpha_{\rm A}} \left(\frac{\partial H}{\partial \alpha_{\rm A}}\right)_{\! T}
-   \left(\frac{\partial G}{\partial \alpha_{\rm A}}\right)_{\! T}\left(\frac{\partial H}{\partial T}\right)_{\! \alpha_{\rm A}}\right]$ is:
\begin{multline*}
 \textstyle 3 (1 + \alpha)\left[1 + \frac{q}{\alpha(1 + \alpha)}\right]
\left(\frac{p}{p_+} - 1\right)
+ 5 (1 + \alpha)\left(1 + \frac{q}{\alpha}\right)\left(1 - \frac{v}{v_+}\right)
 \\
   \textstyle   + 8\left[
 q \left(\frac{3}{2}  +  \frac{T_{\rm A}}{T}\right)\left(\frac{5}{2} + \frac{T_{\rm A}}{T}\right)
 +  q\left(\frac{T_{\rm B} - T_{\rm A}}{T}\right)\frac{Q_{\rm BA}}{\alpha(1 + \alpha) }\right.\\  
 \textstyle  \left.  + \left(4 + \frac{T_{\rm A} + T_{\rm B}}{T} -  \frac{Q_{\rm BA}}{\alpha(1 + \alpha) }\right)Q_{\rm BA} 
- \frac{1}{4}\left(\frac{qT_{\rm A}}{T} + Q_{\rm BA}\right) \right]\left[\frac{T(1 + \alpha)}{T_+(1 + \alpha^+)}  - 1\right]
\\
  \textstyle +  6 (1 + \alpha)\left[1 + \frac{q}{\alpha(1 + \alpha)}\right]
\left(\frac{p}{p_+} - 1\right)
+ 5 (1 + \alpha)\left(1 + \frac{q}{\alpha}\right)\left(1 - \frac{v}{v_+}\right)
 \\
   \textstyle   + 2\left[
 q \left(\frac{3}{2}  +  \frac{T_{\rm A}}{T}\right)\left(\frac{5}{2} + \frac{T_{\rm A}}{T}\right)
 +  q\left(\frac{T_{\rm B} - T_{\rm A}}{T}\right)\frac{Q_{\rm BA}}{\alpha(1 + \alpha) }\right.\\  
  \textstyle  \left.  + \left(4 + \frac{T_{\rm A} + T_{\rm B}}{T} -  \frac{Q_{\rm BA}}{\alpha(1 + \alpha) }\right)Q_{\rm BA} \right]\frac{\left[\beta T_{\rm A}(\alpha_{\rm A}  -\alpha_{\rm A}^+) + (1 - \beta)T_{\rm B}(\alpha_{\rm B} -\alpha_{\rm B}^+)\right] }{T_+(1 + \alpha^+)}\\
 \textstyle + 2 \left(\frac{v}{v_+}\right)\left[\frac{q\left(T_{\rm B} - T_{\rm A}\right)}{T} -  Q_{\rm BA}\right]\left(\frac{p}{p_+} - 1\right)
\frac{Q_{\rm BA}}{1 + \alpha}
\end{multline*}
which implies \eqref{eq:dG/dT final}.
\section{Isentropes}\label{sec:isentrope}
In $(p,u,T)$ coordinates, we have observed by Remark \ref{nb:isentropes} that:
the thermodynamic part of an integral curve is $\eta =$ const. for $1,2$-characteristic directions and $p =$ const. for $0$-characeristic field.
A curve $\eta =$ const.  in $(T, \alpha_{\rm A})$ plane is called an {\it isentrope\/}.
%
%
%
\begin{theorem}\label{eq:alpha(T)}
An isentrope $\eta = \eta_0$ is the graph of a differentiable function $\alpha_{\rm A} = \alpha_{\rm A}(T)$ defined on  $T \in (0, \infty).$
\end{theorem}
\begin{proof}
Derivative of $\eta$ with respect to $\alpha_{\rm A}$ takes a form
  \begin{align}
 	\left(\frac{\partial \eta}{\partial \alpha_{\rm A}}\right)_{\! T}
	&=  \beta\left[\frac{1}{\beta \alpha_{\rm A} + (1 - \beta)\alpha_{\rm B}}  +  \frac{1 + \alpha_{\rm A}}{q_{\rm A}} + \left(\frac{5}{2} + \frac{T_{\rm A}}{T}\right)\right]\nonumber\\
	&\hspace{2.5ex}  +(1 - \beta)\left[\frac{1}{\beta \alpha_{\rm A} + (1 - \beta)\alpha_{\rm B}}  +  \frac{1 + \alpha_{\rm B}}{q_{\rm B}} + \left(\frac{5}{2} + \frac{T_{\rm B}}{T}\right)\right]\left(\frac{\partial \alpha_{\rm B}}{\partial \alpha_{\rm A}}\right)_{\! T}.\label{eq:deta/dalpha}
 \end{align}
By \eqref{eq:dalpha_B/dT} we find  $\left(\frac{\partial \eta}{\partial \alpha_{\rm A}}\right)_{\! T} > 0.$ We have also
$$
 	\left(\frac{\partial \eta}{\partial T}\right)_{\! \alpha_{\rm A}}
	=  - \frac{T_{\rm A}}{T^2}(1 + \alpha) + (1 - \beta)\left[\frac{1}{\alpha} + \left(\frac{5}{2} + \frac{T_{\rm B}}{T}\right) \right] \frac{(T_{\rm B} - T_{\rm A})q_{\rm B}}{T^2} \label{eq:deta/dT}.
$$

\par
First we shall prove that $\alpha_{\rm A}$ is a function of $T$ defined on $(0, \infty).$ Let us fix any $T$ in $(0, \infty).$ It follows from \eqref{eq:compatibility small}
$
  \alpha_{\rm B} \sim \frac{\mu_{\rm A}}{\mu_{\rm B}}\alpha_{\rm A}e^{-\frac{T_{\rm B} - T_{\rm A}}{T}} \quad \text{as} \quad \alpha_{\rm A} \to 0.
$ 
Hence
$$
   \log \alpha_{\rm B} \sim \log \alpha_{\rm A} + O(1), \quad
   \log \left[\beta \alpha_{\rm A} + (1 - \beta)\alpha_{\rm B}\right] 
   \sim \log \alpha_{\rm A} + O(1).
$$
Thus we find
$
  \eta( \alpha_{\rm A},  \alpha_{\rm B}, T) \sim 2 \log \alpha_{\rm A} + O(1) \to -\infty.
  $
  for  $\alpha_{\rm A} \to 0.$
\par
On the othe hand, when $\alpha_{\rm A} \to 1,$ obviously $\alpha_{\rm B} \to 1$ and
$$
 \eta( \alpha_{\rm A},  \alpha_{\rm B}, T) 
\sim - 2\beta  \log (1 - \alpha_{\rm A})  - 2(1 - \beta)  \log (1 - \alpha_{\rm B}) 
+ \beta \left(\frac{5}{2} + \frac{T_{\rm A}}{T}\right) + (1 - \beta)\left(\frac{5}{2} + \frac{T_{\rm B}}{T}\right) \to \infty.  
$$
Consequently,  since $\left(\frac{\partial \eta}{\partial \alpha_{\rm A}}\right)_{\! T} > 0,$  we conclude that there is a unique sigle root $ \alpha_{\rm A}(T)$ such that
$$
\eta( \alpha_{\rm A}(T),  \alpha_{\rm B}(T), T) = \eta_0\quad \text{and} \quad  0 < \alpha_{\rm A}(T) < 1.
$$
\end{proof}
\par
Next we consider the behaviour as $T \to 0.$
\begin{proposition}\label{prop:isentrope T to 0} 
If $T \to 0,$ then $\alpha_{\rm A} \to 0 $  and
  $
  	\alpha_{\rm A}  \sim  \frac{1}{\sqrt{\beta}}\left( \frac{\mu_{\rm A}}{\mu_{\rm B}}\right)^{-\frac{1}{2}(1-\beta)}e^{-\frac{T_{\rm A}}{2T} + \frac{\eta_0}{2}}.
  $
\end{proposition}
\begin{proof}
If $\alpha_{\rm A} \geq c > 0$ or $\alpha_{\rm B} \geq c > 0,$ then $\eta \to \infty$ which is contradiction. Hence $\alpha_{\rm A}, \alpha_{\rm B} < c.$ By compatibility condition, we have
 $$
    \frac{1 - \alpha_{\rm B}}{\alpha_{\rm B}}
    =  \frac{\mu_{\rm B}(1 - \alpha_{\rm A})}{\mu_{\rm A}\alpha_{\rm A}}e^{\frac{1}{T}(T_{\rm B} - T_{\rm A})}.
 $$
If $\alpha_{\rm A} \leq c < 1,$ then by setting $T \to 0$ in the above equation, we have $\alpha_{\rm B} \to 0.$  
Suppose that $ \alpha_{\rm A} \geq c'\ (c > c' >0).$ Since 
 $
 	\log \alpha_{\rm B} = \log (1 - \alpha_{\rm B}) -  \log \frac{1 - \alpha_{\rm A}}{\alpha_{\rm A}} - \frac{1}{T}(T_{\rm B} - T_{\rm A}) + O(1),
 $
we have by substituting the above into the expression of the entropy
 \begin{align*}
 \eta  &= \log \left[\beta \alpha_{\rm A} + (1 - \beta)\alpha_{\rm B}\right] + \log \alpha_{\rm A} - (1 + \beta) \log (1 - \alpha_{\rm A})  
      - (1 - \beta)\log (1 - \alpha_{\rm B})\\
  & + \frac{T_{\rm A}}{T} + \beta \left(\frac{5}{2} + \frac{T_{\rm A}}{T}\right)\alpha_{\rm A} + (1 - \beta)\left(\frac{5}{2} + \frac{T_{\rm B}}{T}\right) \alpha_{\rm B}
 + O(1).
    \end{align*}
   Clearly $\log (1 - \alpha_{\rm A})  < 0,$ $\log (1 - \alpha_{\rm B})  < 0,$ and
   $\log \left[\beta \alpha_{\rm A} + (1 - \beta)\alpha_{\rm B}\right] \geq -c$
   for some $c  > 0.$ Thus we find  $\eta \to \infty$ as $T \to 0,$ which is also contradictory. Thus $\alpha_{\rm A} \to 0$ and the first part of proposition is proved.
   \par
   It follows from \eqref{eq:compatibility small} $\log \alpha_{\rm B} \sim \log \frac{\mu_{\rm A}}{\mu_{\rm B}} + \log \alpha_{\rm A} - \frac{T_{\rm B} - T_{\rm B}}{T}$ as $\alpha_{\rm A} \to 0$ and
\begin{align*}
   \eta & \sim \log \beta + \log  \alpha_{\rm A} + \beta\left(\log  \alpha_{\rm A} + \frac{T_{\rm A}}{T}\right)
   + (1 - \beta) \left(\log \frac{\mu_{\rm A}}{\mu_{\rm B}} + \log \alpha_{\rm A} - \frac{T_{\rm B} - T_{\rm A}}{T} + \frac{T_{\rm B}}{T}\right) \\
       & \sim \log \beta +  (1 - \beta)\log \frac{\mu_{\rm A}}{\mu_{\rm B}} + 2\log  \alpha_{\rm A} + \frac{T_{\rm A}}{T} \sim \eta_0.
\end{align*}
Thus $\alpha_{\rm A} \sim C e^{-\frac{T_{\rm A}}{T}}$ with a certain constant $C$ satisfying $\log C + \log \beta +  (1 - \beta)\log \frac{\mu_{\rm A}}{\mu_{\rm B}} = \eta_0$ and we obtain the asymptotic formula.
   \end{proof}

\begin{acknowledgements}
  
The author has been a member of the Institute of Liberal Arts and Sciences Osaka Electro-Communication University until 30 March 2019 and acknowledges support from this institute. 
\end{acknowledgements}

  %
  \section*{Compliance with Ethical Standards}
  \begin{description}
  \item[Support from the institute:]
    The author has received  continued  support from   Institute of Liberal Arts and Sciences Osaka Electro-Communication University.
  \item[Conflict of interest:]
    The author declares there are no conflict of interest.
  \item[Ethical approval:]
    This article does not contain any studies with human participants or animals performed by the author.
    \end{description}
%


\end{document}